\documentclass[12pt]{article}
\usepackage{amssymb,latexsym}
\usepackage[cmex10]{amsmath}
\newtheorem{theorem}{Theorem}[section] 
\newtheorem{corollary}[theorem]{\bf Corollary} 
\newtheorem{proposition}[theorem]{\bf Proposition} \newtheorem{definition}[theorem]{\bf Definition} \newtheorem{lemma}[theorem]{\bf Lemma} \newtheorem{example}[theorem]{\bf Example}  
\newtheorem{algorithm}[theorem]{\bf Algorithm} 
\newtheorem{remark}[theorem]{\bf  Remark} \newtheorem{remarks}[theorem]{\bf  Remarks} \newenvironment{proof}{{\bf Proof.}}{\hspace*{\fill}$\blacksquare$\par\vspace{4mm}} 

\def \bt{ \begin{theorem} }
\def \et{ \end  {theorem} }
\def \bl{ \begin{lemma} }
\def \el{ \end  {lemma} }
\def \bp{ \begin{proposition} }
\def \ep{ \end  {proposition} }
\def \bc{ \begin{corollary} }
\def \ec{ \end  {corollary} }
\def \bd{ \begin{definition} }
\def \ed{ \end  {definition} }
\def \bdp{ \begin{definitionprop} }
\def \edp{ \end  {definitionprop} }
\def \bdt{ \begin{definitiontheorem} }
\def \edt{ \end  {definitiontheorem} }
\def \bpr{ \begin{proof} }
\def \epr{ \end  {proof} }
\def \ba{ \begin{algorithm} }
\def \ea{ \end{algorithm} }
\def \be{ \begin{example} }
\def \eex{ \end{example} }
\def \bes{ \begin{examples} }
\def \eexs{ \end{examples} }
\def \br{ \begin{remark} }
\def \er{ \end{remark} }
\def \brs{ \begin{remarks} }
\def \ers{ \end{remarks} }
\def \bpb{ \begin{problem} }
\def \epb{ \end{problem} }

\newcommand{\Ann} {\mathrm{Ann}}
\newcommand{\Min} {\mathrm{MP}}

\newcommand{\mydiv} {\ \mathrm{div}\ }
\newcommand{\ol} {\overline}

\newcommand{\ul} {\underline}
\newcommand{\ee} {\mathrm{e}}

\newcommand{\J} {\mathrm{J}}

\newcommand{\ra} {\rightarrow}
\newcommand{\Ra} {\Rightarrow}

\newcommand{\vv} {\mathrm{v}}
\newcommand{\LC} {\mathrm{L}}

\newcommand{\F}{\mathbb{F}}
\newcommand{\N}{\mathbb{N}}
\newcommand{\Z}{\mathbb{Z}}

\newcommand{\MR} {\mathrm{MR}}
\newcommand{\tfae} {The following are equivalent:}
\title{B\'{e}zout Identities Associated to a Finite Sequence.}
\author{G. H. Norton\\\\ Department of Mathematics\\ University of Queensland, Brisbane, Queensland 4072, Australia.\thanks{
Email:ghn@maths.uq.edu.au}}
\begin{document}
\maketitle
\begin{abstract}
We consider finite sequences $s\in D^n$ where $D$ is a 
commutative, unital, integral domain. We prove three sets of identities (possibly with repetitions), each involving $2n$  polynomials  associated to  $s$. The right-hand side of these identities is a recursively-defined (non-zero) 'product-of-discrepancies'. There are implied iterative algorithms (of quadratic complexity) for the left-hand side coefficients; when the ground domain is factorial, the identities are in effect B\'ezout identities.

We give a number of applications. Firstly, a new (quadratic) algorithm to compute  B\'ezout coefficients over a field, which compares favourably with the extended Euclidean algorithm.  We show that  the successive output  polynomials of the Berlekamp-Massey algorithm  either coincide or are relatively prime.  A second application concerns  sequences with perfect linear complexity profile. We give new characterisations of them in terms of minimal polynomials and a simpler proof of a theorem of the characterisation of binary sequences with perfect linear complexity due to Wang and Massey using the third set of identities. 

Another application is to annihilating polynomials which do not vanish at zero and have minimal degree. We simplify and extend an algorithm of Salagean to sequences over $D$. First we prove a lower bound lemma which was stated without proof for sequences over a field. This lemma and an easy extension of the author's minimal polynomial algorithm yields such annihilators. The first set of identities allows us to remove a test. In fact, we compute minimal realisations and give corresponding identities. We also construct these annihilators by extending the sequence by one term, give corresponding minimal polynomial identities and prove a characterisation (stated without proof by Salagean over a field). In the Appendix, we give an alternative proof of the lower bound lemma using reciprocal annihilators and apply it to the complexity of reverse sequences. This gives a new proof of a theorem of Imamura and Yoshida on the linear complexity of reverse sequences, initially proved using Hankel matrices over a field and now valid for sequences over a factorial domain.
\end{abstract}
\section{Introduction}
\subsection{Overview}

Let $\F$ be a field and $(u,u_2)\in \F[x]^2$ not both zero. It is well-known that there are B\'{e}zout coefficients  $(f,f_2)\in \F[x]^2$ such that $f\cdot u+f_2\cdot u_2=\gcd(u,u_2)$, known as a B\'{e}zout\footnote{Etienne B\'{e}zout (1730-1783); according to \cite[p. 115]{Tignol}, the integer case is due to Claude Bachet de M\'eziriac (1581-1638): "Probl\`{e}mes plaisans et d\'{e}lectables qui se font par les nombres", 1624.} identity.
 (Of course, it easy to check that if $(f,f_2)$ are B\'ezout coefficients for $(u,u_2)$, then so are $(f,f_2)+h\cdot(-u_2,u)/\gcd(u,u_2)$ for any $h\in \F[x]$.) As far as we know,   B\'{e}zout coefficients are always computed using the extended Euclidean algorithm. 

In this  paper, we are interested in the case when $(u,u_2)$ arises from a finite sequence. For example, this pair could be the pair  of polynomials used to find a minimal one, the successive minimal polynomials or the rational function approximating the 'generating function' of the sequence. All of these are typically relatively prime polynomials. These pairs are all computed using the minimal polynomial algorithm due to the author \cite{N95b}, \cite{N99b}; for a succinct overview, see \cite{N09d}. A natural question is whether there are associated minimal polynomial  identities. We answer this in the affirmative for each pair of polynomials and give some applications.

Thus let $s=(s_1,\ldots,s_n)\in D^n$ denote a sequence of $n$ terms from a
commutative unital integral domain $D$, a 'domain' for short. We prove theorems giving $n$ minimal polynomial identities (possibly with repetitions) associated to  $s$, Theorems \ref{fg} and \ref{numu}. It turns out the right-hand sides  are equal to the same recursively-defined 'product-of-discrepancies' $\nabla_n\in D\setminus\{0\}$ of Definition \ref{nabla}. There is an implied iterative algorithm of quadratic complexity, Algorithm \ref{MRplusmult}.  (When the ground domain is factorial, the identities are in effect B\'ezout identities as the minimal polynomials are relatively prime.)

We give a number of applications:  the successive 'connection polynomials' of the Berlekamp-Massey algorithm (as developed in \cite{N10a}) either coincide or are relatively prime when $D$ is factorial, Corollary \ref{relprime} and a new (quadratic) algorithm to compute  B\'ezout coefficients, which compares favourably with the extended Euclidean algorithm; see Section 7. Another application concerns  sequences with perfect linear complexity profile. In Section 8 we give new characterisations of them in terms of minimal polynomials and a simpler proof of a theorem of Wang and Massey characterising binary sequences with perfect linear complexity, \cite{Wangpdf}. 

Section 9 gives applications to annihilating polynomials which do not vanish at zero. We  simplify and extend to $D$ results of \cite{Salagean}, which used  definitions of a linear recurring sequence, annihilating and minimal polynomial for finite sequences over $\F$ (with first term $s_0$) equivalent to the non-standard definitions of   \cite{N95b}, \cite{N99b}\footnote{We note that   \cite{N95b} and \cite{N99b} were cited in \cite{NS-key}.}. First we establish a lower bound lemma which was stated without proof in \cite{Salagean}. The resulting corollary and algorithm do not require any characterisations of minimal polynomials.  In particular, the derivation of  Algorithm \ref{rewrite2} is similar to \cite{Ma69} and \cite[Algorithm 4.2]{N95b}.  Furthermore, our first identity allows us to eliminate a test used in \cite[Algorithm 3.2]{Salagean}. We also construct annihilators that do not vanish at zero by extending the given sequence by one term and give corresponding minimal polynomial identities, Theorem \ref{dge0}. We characterise these annihilators of minimal degree in Theorem \ref{char}; a version over $\F$ was stated without proof in \cite{Salagean}. 

In the Appendix (Section 10), we give an alternative proof of Lemma \ref{x|min} using reciprocal annihilators and apply it to the complexity of reverse sequences. This gives a new proof of \cite[Theorem 3, p. 149]{IY} on the linear complexity of reverse sequences, initially proved using Hankel matrices over a field and now valid when $D$ is factorial.

 A preliminary version of this work was presented in May 2010 at Equipe SECRET, Centre de Recherche, INRIA Paris-Rocquencourt, whom the author  thanks for their hospitality.
\subsection{In More Detail}
Our approach is based on the recursive minimal polynomial theorem of \cite{N10a} which is relative to any $\varepsilon\in D$.
This gives the next minimal polynomial $\mu^{(n)}$ using a previous one $\mu^{(n-1)}$ and a $\mu'^{(n-1)}$; informally, $\mu'^{(n-1)}$ is a 'prejump' polynomial of $(s_1,\ldots,s_{n-1})$, i.e. a minimal polynomial immediately preceding a jump in the linear complexity $\LC_{n-1}=\deg(\mu^{(n-1)})$ --- we are ignoring the initial cases. We prove inductively that $\mu^{(n)}$ and $\mu'^{(n)}$ are always relatively prime (when $D$ is factorial) and exhibit polynomials $f^{(n)}$ and $f_2^{(n)}$ such that for any domain $D$
$$f^{(n)}\cdot\mu^{(n)}+f_2^{(n)}\cdot\mu'^{(n)}=\nabla_n.$$
We apply this to obtain an identity for $\mu^{(n-1)}$ and $\mu^{(n)}$, Proposition \ref{nextid}.

Before discussing our third identity, we need to define another polynomial $\mu_2^{(n)}$ derived from $\mu^{(n)}$ and ${s}$; $\mu_2^{(n)}\in D[x]$ is the {\em polynomial part of the Laurent product} $$\mu^{(n)}\cdot (s_nx^{-n}+\cdots +s_1x^{-1}).$$
In an earlier paper \cite{N95b}, we inductively constructed  $\ol{\mu}^{(i)}=(\mu^{(i)},\mu_2^{(i)})\in D[x]^2$ for $1\leq i\leq n$ (possibly with repetitions) such that

(i)  $\deg(\mu_2^{(i)})<\deg(\mu^{(i)})$

(ii)  $\mu^{(i)}\cdot (s_ix^{-i}+\cdots +s_1x^{-1})-\mu_2^{(i)} \equiv 0\bmod(x^{\deg(\mu^{(i)})-i-1})$ 

(iii) $\mu^{(i)}$ is a solution to (ii) of minimal degree.

\noindent See Algorithm \ref{MRplusmult},  which does not use scalar or polynomial division. We called the pair
$\ol{\mu}^{(n)}$ a 'minimal realisation' of $s$.
If $\mu^{(n)}$ is monic, this gives a rational approximation of $s_1x^{-1}+\cdots +s_nx^{-n}$ with denominator of minimal degree. 
When $D$ is factorial, minimality easily implies that 
$\mu^{(n)}$ and $\mu_2^{(n)}$ are relatively prime, so it is natural to ask what their 'coefficients' are. Our third identity is
$$-\mu_2'^{(n)}\cdot\mu^{(n)}+\mu'^{(n)}\cdot\mu_2^{(n)}=\nabla_n.$$
 Algorithm \ref{MRplusmult} already computes these coefficients,  so they can be read off at the end of the $n^{th}$ iteration at no extra cost. 
The degrees of the coefficients are bounded above by $\LC_{n-1}$.  We also show that if $D$ is a field, these coeffients are valid  for {\em any} minimal realisation of $s$, Corollary 
\ref{anyMR}. Secondly, our coefficients agree with the coefficients found by the extended Euclidean algorithm (up to a scalar), Proposition \ref{xea}. 

We conclude this Introduction by outlining our algorithm for obtaining the B\'{e}zout identity  for $(u,u_2)\in D[x]^2$ where $0\leq \deg(u_2)\leq \deg(u)=d$, $u$ is monic and $D$ is a principal ideal domain. The steps are: 

 (i) $2d$ subtractions in $D[x^{-1},x]$ to obtain $s=u_2/u (\bmod \ x^{-2d-1})$
 
 (ii) Algorithm \ref{MRplusmult} for the B\'{e}zout coefficients 
 associated to $s$
 
 (iii) two multiplications in $D[x]$ to obtain $\nabla_{2d}\cdot\gcd(u,u_2)$. 
 
This means that we can use Algorithm \ref{MRplusmult} whenever  the extended Euclidean algorithm is applied e.g.  to compute B\'{e}zout coefficients in $\F[x,y]$ for example.
B\'{e}zout identities also arise  in the context of the key equation of Coding Theory, see e.g. \cite{McE02}, \cite{N95c}, \cite{NS-key} and it seems likely that the algorithm of \cite[Section 4]{AM} is closely related to the minimal realisation algorithm of \cite{N95b}; see also \cite[Section 8]{N99b}.

\section{The Recursive Minimal Polynomial Theorem}
\subsection{Notation}
Our notation is continued from \cite{N10a}:  $\N=\{1,2,\ldots\}$, $n\in\N$ and $D$ is  a commutative, unital integral domain with $1\neq 0$, or 'domain' for short. For any set $S$ containing 0, $S^\times=S\setminus\{0\}$. We say that $f\in D[x]^\times$ is monic if its leading term is 1 and the constant term of $f$ is  $f_0$. The reciprocal of 0 is 0 and for $f\in D[x]^\times$, its reciprocal is $f^\ast(x)=x^{\deg(f)}f(x^{-1})$.  We  often write  $f=x^eg+h$ for $f(x)=x^eg(x)+h(x)$, where $e\in\N$ and $g,h\in D[x]$.

For $f\in D[x]$, $\ol{f}$ denotes an element of $D[x]^2$ with first component $f$ and a certain component  $f_2$ (depending on a sequence) to be specified. We regard
$D[x]^2$ as a $D[x]$-module by multiplication on each
component; e.g. for $e\in\N$, $\Delta\in D$ and $f,f'\in D[x]$
$$x^e\ol{f}+\Delta\cdot\ol{f}'=(x^ef+\Delta\cdot f',x^ef_2+\Delta \cdot f_2').$$
We will also extend $\deg$ to $D[x]^2$ via
$\deg(\ol{f})=(\deg(f),\deg(f_2))$ and write $\gcd\ol{f}$ for $\gcd(f,f_2)$.
If $\ol{f}=(f,f_2)$, we let $$\tilde{f}=(-f_2,f)$$
 We also use $\cdot$ for  the {\em inner product} $\ol{f}\cdot\ol{u}=(fu,f_2u_2)$.

We write $:D[[x^{-1},x]$ for the domain of Laurent series in $x^{-1}$ over $D$ and let $\vv:D[[x^{-1},x]\ra \Z$ be the exponential valuation. We also use $\vv$ denote its restriction to $D[x^{-1},x]$. Thus $\vv:D[x^{-1},x]\ra \Z$ is $\vv(g)=\max\{i: g_i\neq 0\}$, which coincides with the degree function on $D[x]$. We have $\vv(gg')=\vv(g)+\vv(g')$, $\vv(g+g')\leq \max\{\vv(g),\vv(g')\}$ and $\vv(g+g')= \max\{\vv(g),\vv(g')\}$ if $\vv(g)\neq \vv(g)'$.
\subsection{Sequences}
Our non-standard approach to finite sequences was based on regarding linear recurring sequences $s:\N\rightarrow D$ as the torsion submodule of a natural $D[x]$-module,  \cite[Section 2]{N95b},  \cite[Section 3]{N99b}. 
 We begin with Laurent series in $x^{-1}$, $D[[x^{-1}]$ as standard $D[[x^{-1}]$-module i.e. acting on itself via multiplication. This also makes $D[[x^{-1}]$ into a $D[x]$-module. Let $\ul{s}=\sum_{j\geq 1}s_jx^{-j}$. Then for $f\in D[x]$, put
$$f\circ \ul{s}=\sum_{j\geq 1}(f\cdot \ul{s})_{-j}\ x^{-j}.$$
 One checks that $\circ$ makes $x^{-1}D[[x^{-1}]]$ into a $D[x]$-module.
 \bd (Cf. \cite{S05}, \cite{AS}, \cite{Salagean}) We say that $s$ satisfies a {\em linear recurrence relation}  if it is a torsion element i.e. if $\Ann(s)=\{f\in D[x]:\ f\circ \ul{s}=0\}\neq \{0\}$ or  if for some $f\in D[x]^\times$  $(f\cdot \ul{s})_{d-j}=0$ for $d-j\leq -1$ where $d=\deg(f)$ i.e.
$$ f_0s_{j-d}+\cdots+f_ds_ {j}=0\mbox{ for } d+1\leq j.
$$
\ed
 When $f_d=1$, we can write 
$s_{j}=-(f_0s_{ j-d}+\cdots+ f_{ d-1}s_{ j-1})$ for $j\geq d+1$ and $s$ is a {\em linear recurring sequence}.

Now for  $n\in \N$ and $s=(s_1,\ldots,s_n)\in D^n$, let $\overline{s}\in D[x]$ be $\overline{s}(x)=s_1x+\cdots+s_nx^n.$ 
We will also abbreviate $\overline{s}(x^{-1})=s_1x^{-1}+\cdots+s_nx^{-n}$ to $\ul{s}$, so that $\ul{s}_j=s_{\ -j}$ for $-n\leq j\leq -1$. 
In the following definition, multiplication of $f\in D[x]$ and $\ul{s}\in D[x^{-1}]$ is in the domain of $D$-Laurent {\em polynomials} $D[x,x^{-1}]$.
\begin{definition}[Annihilator, annihilating polynomial] \label{anndefn}(\cite[Definition 2.7, Proposition 2.8]{N95b}, Cf.  \cite{AS}, \cite{Salagean})) If $s\in D^n$, then $f\in D[x]$  is an annihilator (or a characteristic polynomial) of $s$ if $f=0$ or $(f\cdot \ul{s})_{d-j}=0$ for $d-n\leq d-j\leq -1$ i.e. 
$$f_0s_{j-d}+\cdots+f_ds_ {j}=0\mbox{ for } d+1\leq j\leq n$$
where $d=\deg(f)\geq 0$, written $f\in \Ann(s)$.
\end{definition}

Any polynomial of degree at least $n$ annihilates $s$, vacuously.  For $1\leq i\leq n$, we write $s^{(i)}$ for $(s_1,\ldots,s_i)$.  The following definition is a functional version of \cite[Definition 2.10]{N95b}.
\begin{definition} [Discrepancy Function] We define $\Delta: D[x]^\times\times D^n\ra D$ by 
$$\Delta(f,s)=(f\cdot\ul{s})_{\deg(f)-n}.$$
\end{definition}
Thus $\Delta(f,s)=\sum_{k=0}^{d}f_k \ s_{n-d+k}$ where $d=\deg(f)$.  Clearly for  $n\geq 2$, $f\in\Ann(s)^\times$  if and only if  $f\in\Ann(s^{(n-1)})^\times$ and $\Delta(s,f)=0$.
If $s\in D^n$ is understood, we write $\Delta_n(f)$ for  $\Delta(f,s)$;  if $f$ is also understood, we simply write $\Delta_n$.
 It is elementary that  if $1\leq i\leq n-1$, then
$(f\cdot\ul{s^{(i)}})_{\deg(f)-i}=(f\cdot\ul{s})_{\deg(f)-i}$.

\section{Minimal Polynomials}\label{mptheory}

\begin{definition}[Minimal Polynomial] \label{mindefn} (\cite[Definition 3.1]{N95b},  Cf. \cite{AS}, \cite{Salagean}) We say that $f\in \Ann(s)$ is a minimal polynomial of $s\in D^n$ if  $$\deg(f)=\min\{\deg(g):\ g\in \Ann(s)^\times\}$$ and let $\Min(s)$ denote the set of minimal polynomials of $s$.
\end{definition}
As any $f\in D[x]$ of degree at least $n$ annihilates $s\in D^n$, $\Min(s)\neq \emptyset$. 
We do not require minimal polynomials to be monic.
For any $d\in D^\times$, $d\in \Min(0,\ldots,0)$; if $s_1\neq 0$ and $\deg(f)=1$  then $f\in \Min((s_1))$ since $D$ has no zero divisors.  

The {\bf linear complexity function} $\LC:D^n\ra\{0\}\cup \N$ is $$\LC(s)=\deg(f)\mbox{ where }f\in\Min(s).$$
We will also write $\LC_n$ for $\LC(s)$ when $s$ is understood and similarly  $\LC_j=\LC(s^{(j)})$ for $1\leq j\leq n$.  For fixed $s$, $\LC$ is clearly a non-decreasing function of $i$.

 \subsection{Exponents}

\begin{definition}[Exponent Function] For $n\geq 1$, let the exponent function $\ee_n:D[x]^\times \times\N\ra \Z$ be given by 
$$\ee_n(f)=n+1-2\deg(f).$$
\end{definition}

 The following lemma is the annihilator analogue of \cite[Lemma 1]{Ma69}.
\begin{lemma}   \label{elegantproof}(\cite[Lemma 5.2]{N95b})  Let $n\geq 2$, $f\in\Ann(s^{(n-1)})^\times$ and $\Delta_n(f)\neq 0$. 

(i)  For any $g\in\Ann(s)^\times$, $\deg(g)\geq n-\deg(f)= \ee_{n-1}(f)+\deg(f)$. 

(ii) If $h\in\Min(s)$ then
$\deg(g)\geq\max\{\ee_{n-1}(h),0\}+\deg(h)$.
\end{lemma}

\section{A Recursive Minimal Polynomial Function}
 We will define a recursive minimal polynomial function $\mu:D^n \ra D[x]$.  But first we need the following function (which assumes that $\mu:D^{n-1} \ra D[x]$ has been defined).
We also set $\Delta_0=1$. 
\subsection{The Index Function}
\begin{definition}[Index Function]\label{indices}
Let $n\geq 1$ and $s\in D^n$. We set $\mu^{(0)}=1$ (so that $\Delta_{1}=\Delta_1(\mu^{(0)})=s_1$) and  $\ee_0=1$. 
Suppose that for $1\leq j\leq n-1$, $\mu^{(j)}\in \Min(s^{(j)})$ has discrepancy $\Delta_{j+1}$ and exponent $\ee_j$. We define  the index function
$$':\{0,\ldots,n\}\rightarrow \{-1,n-1\}$$ 
by $0'=-1$ and for $1\leq j\leq n-1$
$$j'=\left \{\begin{array}{ll}
 (j-1)' &\mbox{ if }\Delta_{j}=0\mbox{ or }(\Delta_{j}\neq 0 \mbox{ and }\ee_{j-1}\leq 0)\\
j-1 & \mbox{ if } \Delta_{j}\neq 0\mbox{ and }\ee_{j-1}>0.
\end{array}
\right.
$$
\end{definition}

\subsection{The Recursive Theorem}
We recall the recursive function  $\mu:D^n\ra  D[x]^\times$
of \cite{N10a}: for all $s\in D^n$, $\mu(s)\in\Min(s)$.
When $s$ is understood, we will write $\mu^{(j)}$ for $\mu(s^{(j)})$.

\begin{definition}[Basis of the Recursion]\label{initialvalues} Recall that $0'=-1$ and $\Delta_0=1$. Let $\varepsilon\in D$ be arbitrary but fixed and $s\in D^n$. We put $\mu^{(-1)}=\mu(s,-1)=\varepsilon$  and $\mu^{(0)}=\mu(s,0)=1$. 
\end{definition}
 Thus the exponent of $\mu^{(0)}$ is $\ee_0=1$ and $\Delta_{1}=s_1$. We also recall some notation from \cite{N10a}:

(i) $\mu'=\mu\circ\ '$ (where $\circ$ denotes composition)

(ii) $\LC'=\deg\circ \mu'$

(iii) $\Delta'=\Delta\circ (+1)\circ\ '\circ (-1)$, where $\pm 1$ have the obvious meanings. 

\noindent Thus $$\mu'^{(j)}=\mu^{(j')},\ \LC'_j=\LC_{j'}\mbox{ and }\Delta'_{j}=\Delta(\mu^{(k)},s^{(k+1)})=(\mu^{(k)}\cdot \ul{s})_{\LC_k-k-1}$$
 where $k=(j-1)'$.

\begin{theorem} (Cf. \cite{Ma69}) 
\label{bit} Let $n\geq 1$ and $s\in D^n$ and assume the initial values of Definition \ref{initialvalues}.  Define  $\mu^{(n)}$ recursively by$$\mu^{(n)}=\left\{\begin{array}{ll}
\mu^{(n-1)}& \mbox{ if } \Delta_n=0\\\\
\Delta'_{n}\cdot x^{\max\{\ee_{n-1},0\}}\ \mu^{(n-1)}-\Delta_{n}\cdot x^{\max\{-\ee_{n-1},0\}}\ \mu'^{(n-1)}
&\mbox{ otherwise.}\end{array}\right.
$$
If $\Delta_n=0$, clearly $\mu^{(n)}\in \Min(s)$, $\LC_n=\LC_{n-1}$ and $\ee_n=\ee_{n-1}+1$. If $\Delta_{n}\neq 0$ then
\begin{tabbing}
\hspace{1cm}\=(i) 
 $\deg(\mu^{(n)})=\max\{\ee_{n-1},0\}+\LC_{n-1}=n'+1-\LC'_{n}$\\ 
\>(ii) $\mu^{(n)}\in \Min(s)$\\
\>(iii) $\ee_n=-|\ee_{n-1}|+1$.
\end{tabbing}
\end{theorem}

\begin{corollary}\label{LL'} (Cf. \cite{Ma69}) For any $s\in D^n$, $\LC_n=n'+1-\LC'_{n}$.
\end{corollary}

\begin{corollary}[Iterative Form of $\mu$]\label{noindices}
Let $n\geq 1$, $s\in D^n$ and $\varepsilon\in D$.  
Assume the initial values of Definition \ref{initialvalues}. For $1\leq j\leq n$, let
$$\mu^{(j)}=\left\{\begin{array}{ll}
\mu^{(j-1)} &\mbox{ if } \Delta_{j}=0\\\\
\Delta'_{j} \cdot  x^{\max\{\ee_{j-1},0\}}  \mu^{(j-1)}- \Delta_{j} \cdot    x^{\max\{-\ee_{j-1},0\}}  \mu^{'({j-1})}&\mbox{ otherwise.}
\end{array}\right.
$$
Then $\mu^{(j)}\in\Min(s^{(j)})$. Further, if $\Delta_j=0$, then $\mu^{'(j)}=\mu^{'(j-1)}$,  $\Delta'_{j+1}=\Delta'_{j}$ and $\ee_j=\ee_{j-1}+1$.
If $\Delta_j\neq 0$ then 
\begin{tabbing}\hspace{1cm}\= 
(a) if $\ee_{j-1}\leq 0$ then $\mu^{'(j)}=\mu^{'(j-1)}$ and $\Delta'_{j+1} =\Delta'_{j}$ \\
\>(b) but if $\ee_{j-1}>0$ then $\mu^{'(j)}=\mu^{(j-1)}$ and $\Delta'_{j+1} =\Delta_{j}$ \\
\>(c) $\ee_j=-|\ee_{j-1}|+1$.
\end{tabbing}
\end{corollary}
In other words, when  $\Delta_{j}\neq 0$,
$$\mu^{(j)}= \left\{\begin{array}{lll}

\Delta'_j \cdot  \mu^{(j-1)}- \Delta_{j} \cdot    x^{-e}  \mu'^{(j-1)}	& \mbox{if } e=\ee_{j-1}\leq 0\\\\
\Delta'_j \cdot  x^{+e}\mu^{(j-1)}- \Delta_{j}  \cdot  \mu'^{(j-1)} &\mbox{otherwise.}
\end{array}
\right.
$$
See \cite{N10a} for the derivation of the following algorithm from Corollary \ref{noindices}.
 \begin{algorithm}[Iterative minimal polynomial]   \label{rewrite}\ 
\begin{tabbing}

\noindent Input: \ \ \=$n\geq 1$, $\varepsilon\in D$ and $s=(s_1,\ldots,s_{n})\in D^n$.\\

\noindent Output: \>$\mu\in\Min(s)$.\\\\

\{$e := 1$;\ $\mu':=\varepsilon$;\ $\Delta':=1$;
$\mu  :=  1;\ $\\
{\tt FOR} \= $j = 1$ {\tt TO }$n$\\
    \> \{$\Delta    :=  \sum_{k=0}^{\frac{j-e}{2}} \mu_k \  s_{k+\frac{j+e}{2}};$ \\
   \> {\tt IF } $\Delta  \neq  0$ {\tt THEN }\{{\tt IF } $e\leq 0$  \=   						{\tt THEN } $\mu   :=  \Delta'\cdot \mu-\Delta\cdot  x^{-e} \mu'$;\\\\
  \>                 \>    {\tt ELSE} \{\=$t :=  \mu$;\\
  \>                 \>           \> $\mu  :=   \Delta'\cdot x^e\mu-\Delta\cdot  \mu'$;\\
  \>                 \>            \> $\mu':= t$; \ $\Delta':= \Delta$;\\ 
    \>                 \>            \>$e := -e$\}\}\\
  \> $e  := e+1$\}\\
{\tt RETURN}$(\mu)$\}
\end{tabbing} 
\end{algorithm}

\section{Minimal Realisations for $D^n$}

We recall the definition of a minimal realisation from \cite{N95b}; when $\mu\in\Min(s)$ is monic, this amounts to rational approximation of $\ul{s}$ by a rational function $\mu_2/\mu$ for a certain $\mu_2\in D[x]$ with $\deg(\mu_2)<\deg(\mu)$ and $\deg(\mu)$ minimal.

Recall that  $\vv:D[x^{-1},x]\ra \Z$ is $\vv(g)=\max\{i: g_i\neq 0\}$.  Let $f\in D[x]^\times$ and $d=\deg(f)$.
Then for $s\in D^n$, $\vv(\ul{s})\leq -1$ and 
$$f\cdot \ul{s}=\sum_{j=-n}^{\vv(\ul{s})+d} (f\cdot \ul{s})_j\ x^j=\left(\sum_{j=-n}^{d-n-1} +\sum_{j=d-n}^{-1} +\sum_{j=0}^{\vv(\ul{s})+d}\right)(f\cdot \ul{s})_j\ x^j$$ 
\bd For fixed $s\in D^n$ and $f\in D[x]$, $f_2=f_2(s)\in D[x]$ is the $D[x]$-summand of $f\cdot\ul{s}$:
$$f_2(x)=
\sum_{j=0}^{\vv(\ul{s})+\deg(f)} (f\cdot \ul{s})_j\ x^j$$ 
and $\ol{f}=(f,f_2)\in D[x]^2$. 
\ed 
Thus  $f\cdot\ul{s}=(f\cdot\ul{s}-f_2)+f_2$,  and  $\deg(f_2)=\vv(\ul{s})+\deg(f)< \deg(f)$. If $f$ is monic, $\vv(f)\vv(f_2/f)=\vv(f_2)$ and $\vv(f_2/f)=\vv(f_2)-\vv(f)<0$ i.e. $\frac{f_2}{f}\in x^{-1}D[[x^{-1}]]$.  
The next result is immediate, by construction.
\begin{proposition} If $s\in D^n$, $f\in \Ann(s)$ if and only if $f=0$ or $\vv(f\cdot \ul{s}-f_2)<\deg(f)-n$.
\end{proposition}
Thus if $f\in\Ann(s)$ is monic, $\vv(\ul{s}-f_2/f)=\vv(f\cdot\ul{s}-f_2)-\vv(f)
\leq \vv(f\cdot\ul{s}-f_2)\leq \deg(f)-n-1$ and $$\ul{s}-\frac{f_2}{f}\equiv 0 \bmod (x^{\deg(f)-n-1}).$$ 

\begin{definition}[Minimal Realisation](\cite[Sections 2,3]{N95b}) We will say that $\ol{f}\in D[x]^2$ {\em realises} $s$ if either (i) $f=0$ or (ii) $f\neq 0$, $\deg(f_2)<\deg(f)$ and $\vv(f\cdot \ul{s}-f_2)<  \deg(f)-n$. Further, $\ol{f}$ is a {\em minimal} {\em realisation} of $s$ if 
$\ol{f}$ realises $s$ and $f\in \Min(s)$, written $\ol{f}\in \mathrm{MR}(s)$.
\end{definition}
Then if $f\neq 0$, $\ol{f}\in \mathrm{MR}(s)$ if and only if $f\in\Min(s)$.

For $\mu^{(n)}$ as in Theorem \ref{bit}, we write $\mu_2^{(n)}$ for the polynomial part of $\mu^{(n)}\cdot\ul{s}$
$$\mu_2^{(n)}(x)=\sum_{k=0}^{\vv(\ul{s})+d} (\mu^{(j)}\cdot \ul{s})_k\ x^k.$$
By definition, if $\mu_2^{(j)}\neq 0$, then $\deg(\mu_2^{(j)})=\vv(\ul{s^{(j)}})+\deg(\mu^{(j)})<\deg(\mu^{(j)})$
 and it is easy to check that
 $\deg(\mu_2^{(j-1)})\leq \deg(\mu_2^{(j)})$, where $1\leq j\leq n$. If $\Delta_n=0$, then we obviously have $\mu_2^{(n)}=\mu_2^{(n-1)}$.
 
 It is convenient to separate out the following special case before extending Theorem \ref{bit} to minimal realisations.
 \bp \label{mu2} Let $n\geq 1$, $s\in D^n$, $\LC_{n-1}=0$ and $\Delta_n\neq 0$. If 
 $\mu_2^{(-1)}=-1$ and $\mu_2^{(0)}=0$, then 
  \begin{eqnarray}
\ol{\mu}^{(n)}=
\Delta'_n\cdot x^{\max\{\ee_{n-1},0\}}\ \ol{\mu}^{(n-1)}-\Delta_{n}\cdot x^{\max\{-\ee_{n-1},0\}}\ \ol{\mu}'.
 \end{eqnarray}
 \ep
 \bpr We have to show that the second components of both sides agree. It is clear that   $s$ has precisely $n-1$ leading zeroes, which gives the following data:  $\Delta_j=0$ for $1\leq j\leq n-1$, $\mu^{(j)}=1$ for $0\leq j\leq n-1$ and $\mu_2^{(n-1)}=0$, $\ee_{n-1}=n$, $(n-1)'=-1$ and $\Delta_n=s_n$, $\Delta'_n=1$. Substituting these gives\, 
 $\mu_2^{(n)}=s_n=\Delta_n\cdot x^n\ \mu_2^{(n-1)}-\Delta_n\cdot\mu'^{(n-1)}_2$, as claimed.
 \epr
An iterative version of the following result (without the value of $\deg(\mu_2^{(n)})$) was proved in \cite{N95b}, \cite{N99b}. For completeness, we prove the recursive form.
\begin{theorem} [Recursive MR] \label{bimrt}(Cf. \cite{N95b}) 
 Let $n\geq 1$ and $s\in D^n$. Assume the initial values of Definition \ref{initialvalues} and set $\mu_2^{(-1)}=-1$ and $\mu_2^{(0)}=0$.  Define  $\ol{\mu}^{(n)}$ recursively by$$\ol{\mu}^{(n)}=\left\{\begin{array}{ll}
\ol{\mu}^{(n-1)}& \mbox{ if } \Delta_n=0\\\\
\Delta'_{n}\cdot x^{\max\{\ee_{n-1},0\}}\ \ol{\mu}^{(n-1)}-\Delta_{n}\cdot x^{\max\{-\ee_{n-1},0\}}\ \ol{\mu}'^{(n-1)}
&\mbox{ otherwise.}\end{array}\right.
$$
If $\Delta_n=0$, clearly $\ol{\mu}^{(n)}\in \MR(s)$, $\LC_n=\LC_{n-1}$ and $\ee_n=\ee_{n-1}+1$. If $\Delta_{n}\neq 0$ then
\begin{tabbing}
\hspace{1cm}\=(i) 
 $\deg(\mu^{(n)})=\max\{\ee_{n-1},0\}+\LC_{n-1}=n'+1-\LC'_{n}$\\ 
 \> \hspace{0.6cm}and if $\mu_2^{(n)}\neq 0$ then
 $\deg(\mu_2^{(n)})=\max\{\ee_{n-1},0\}+\deg(\mu_2^{(n-1)})$\\ 

\> (ii) $\ol{\mu}^{(n)}\in \MR(s)$\\

\>(iii) $\ee_n=-|\ee_{n-1}|+1$.\\

\end{tabbing}
\end{theorem}

\bpr Put $v=\vv(\ul{s})$, $\mu=\mu^{(n-1)}$, $\LC=\LC_{n-1}$, $e=\ee_{n-1}$,  $\mu'=\mu'^{(n-1)}$, $\LC'=\LC'_{n-1}$, $\Delta'=\Delta'_n$ and $\Delta =\Delta_n$. From Theorem \ref{bit}, we have 
$${\mu}^{(n)}=
\Delta'\cdot x^{\max\{e,0\}}\ {\mu}-\Delta\cdot x^{\max\{-e,0\}}\ {\mu'}.$$
First we show that
${\mu_2}^{(n)}=
\Delta'\cdot x^{\max\{e,0\}}\ {\mu_2}-\Delta\cdot x^{\max\{-e,0\}}\ {\mu'_2}.$
For $e\leq 0$
\begin{eqnarray*}
\mu_2^{(n)}&=&\sum_{j=0}^{v+\LC_n}(\mu^{(n)}\cdot\ul{s})_j\ x^j
=\Delta'\cdot\mu_2-\Delta\cdot\sum_{j=0}^{v+\LC}(x^{-e}\mu'\cdot\ul{s})_j\ x^j\\
\end{eqnarray*}
since $\LC_n=\LC$. The second summation is 
$\sum_{k=e}^{v+\LC+e}(\mu'\cdot\ul{s})_k\ x^{k-e}$
\ i.e. 
$$x^{-e}\left(\sum_e^{-1}+\sum_0^{v+\LC'}+\sum_{v+\LC'+1}^{v+\LC+e}\right)(\mu'\cdot\ul{s})_k\ x^k=A+x^{-e}\mu_2'+B$$
say. Now for $n\geq 1$, $(n-1)'\leq n-2\leq n-2+\LC'$.
From Theorem \ref{bit}(i), $\LC=(n-1)'+1-\LC'$ which gives $\LC'-(n-1)'\leq e$ since $\LC\leq n-1$ and so $A=0$. Likewise, $\LC+e=n-\LC=n-1-(n-1)'+\LC'\geq \LC'+1$ and so
 $B=0$. Thus $\mu_2^{(n)}=\Delta'\cdot \mu_2-\Delta\cdot x^{-e}\ \mu'_2$ if  $e\leq 0$.
 
 For $e<0$, we have
\begin{eqnarray*}
\mu_2^{(n)}&=&\sum_{j=0}^{v+n-\LC}\left(\Delta'\cdot x^e\mu\cdot\ul{s}
-\Delta\cdot\mu'\cdot\ul{s}\right)_j\ x^j\\
&=&\Delta'\cdot\sum_{j=0}^{v+n-\LC}(x^e\mu\cdot\ul{s})_j\ x^j
-\Delta\cdot\sum_{j=0}^{v+n-\LC}(\mu'\cdot\ul{s})_j\ x^j\\
&=&\Delta'\cdot x^e\left(\sum_{k=-e}^{-1}+\sum_{k=0}^{v+\LC}\right)(\mu\cdot\ul{s})_k\ x^k
-\Delta_n\cdot\sum_{j=0}^{v+\LC'}(\mu'\cdot\ul{s})_j\ x^j\\
&=&\Delta'\cdot x^e\sum_{k=-e}^{-1}(\mu\cdot\ul{s})_k\ x^k+\Delta'\cdot x^e\mu_2
-\Delta\cdot\mu_2'
\end{eqnarray*}
since  $v+n-\LC-e=v+\LC$ and $v+n-\LC>v+\LC\geq v+\LC'$. 
From Proposition \ref{mu2}, we can assume that $\LC\geq 1$. This implies that $\LC-(n-1)\leq -e$  and so the first summand vanishes, giving
 $\mu_2^{(n)}=\Delta'\cdot x^e\mu_2-\Delta\cdot\mu_2'
$ if  $e>0$.\\

Next we prove that if $\mu_2^{(n)}\neq 0$ then $\deg(\mu_2^{(n)})=\max\{e,0\}+\deg(\mu_2)$.  
Since  $\mu_2\neq 0$, $s^{(n-1)}\neq (0,\ldots,0)$ and so $v'=\vv(s_{n-1}x^{1-n}+\cdots + s_1x^{-1})=v$. Thus $\deg(\mu_2^{(n)})=v+\LC_n=v'+\max\{\LC,n-\LC\}
=v'+\LC=\deg(\mu_2)$ if $e\leq 0$. If $e>0$, then $\LC_n=n-\LC=\LC+e$ and $\deg(\mu_2^{(n)})=v+(\LC+e)=(v'+\LC)+e =\deg(\mu_2)+e$.
\epr

\begin{corollary}[Iterative Form of $\ol{\mu}$]\label{olnoindices}
Let $n\geq 1$, $s\in D^n$ and $\varepsilon\in D$.  
Assume the initial values of Definition \ref{initialvalues}. For $1\leq j\leq n$, let
$$\ol{\mu}^{(j)}=\left\{\begin{array}{ll}
\ol{\mu}^{(j-1)} &\mbox{ if } \Delta_{j}=0\\\\
\Delta'_{j} \cdot  x^{\max\{\ee_{j-1},0\}}\  \ol{\mu}^{(j-1)}- \Delta_{j} \cdot    x^{\max\{-\ee_{j-1},0\}}\  \ol{\mu}^{'({j-1})}&\mbox{ otherwise.}
\end{array}\right.
$$
Then $\ol{\mu}^{(j)}\in\MR(s^{(j)})$. Further, if $\Delta_j=0$, then $\ol{\mu}'^{(j)}=\ol{\mu}'^{(j-1)}$,  $\Delta'_{j+1}=\Delta'_{j}$ and $\ee_j=\ee_{j-1}+1$.
If $\Delta_j\neq 0$ then 
\begin{tabbing}\hspace{1cm}\= 
(a) if $\ee_{j-1}\leq 0$ then $\ol{\mu}'^{(j)}=\ol{\mu}'^{(j-1)}$ and $\Delta'_{j+1} =\Delta'_{j}$ \\
\>(b) but if $\ee_{j-1}>0$ then $\ol{\mu}'^{(j)}=\ol{\mu}^{(j-1)}$ and $\Delta'_{j+1} =\Delta_{j}$ \\
\>(c) $\ee_j=-|\ee_{j-1}|+1$.
\end{tabbing}
\end{corollary}

For future use, we recall:
\bl \label{f2} Let  $s\in D^n$.

(i)(\cite[Proposition 2.3]{N95b}, \cite[Lemma 6.1]{N99b}) If  $g\in D[x]$, $h\in\Ann(s)$ and 
$\deg(g)+\deg(h)\leq n$ then $(gh)_2=gh_2$.

(ii) If $\deg(f)\leq n-\LC$ and $\deg(f')\leq\LC-\LC'$ then $$(f\cdot\mu+f'\cdot\mu')_2=f\cdot\mu_2+f'\cdot\ \sum_{j=0}^{\LC'-1}\ (\mu'\cdot \ul{s})_j\ x^j.$$
\el
\bpr (i) Let $H=h\cdot\ul{s}-h_2$, $d=\deg(g)$ and $e=\deg(h)$. Since $h\in\Ann(s)$, $h\ul{s}=H+h_2$ and $(gh)\cdot\ul{s}=g(h\cdot\ul{s})=gH+gh_2$. As $\vv(gH)=\deg(g)+\vv(H)\leq d-n+e-1\leq -1$,  $(gh)_2=gh_2$.

(ii) Let $d=\deg(f)$. Then from Part (i),  
$$(f\cdot\mu+f'\cdot\mu')_2= (f\cdot\mu)_2+\sum_{j=0}^{d+\LC}(f'\mu'\cdot\ul{s})_j\ x^j=f\cdot\mu_2+\sum_{j=0}^{d+\LC}(f'\mu'\cdot\ul{s})_j\ x^j$$
and so it suffices to evaluate the second summand. We have
$$\mu'\cdot \ul{s}=\left(\sum_{j=-n}^{-n+\LC'-1}+\sum_{j=-n+\LC'}^{-1}+\sum_{j=0}^{\LC'-1}\right)\ (\mu'\cdot \ul{s})_j\ x^j$$
Now $\deg(f')-n+\LC'-1\leq \LC-n-1\leq -1$ since $\LC\leq n$ and the middle summand is zero. Finally, if $\LC'\leq j\leq d+\LC$ then $(\mu\cdot\ul{s})_j=0$ since $\deg(\mu')=\LC'$ and $\ul{s}_{j-\LC'}=0$.  
\epr

\begin{proposition}\label{MRprop} 
(\cite[Proposition 4.13]{N95b})  Let $n\geq 1$, $s\in D^n$,  $f'\in D[x]$ and  $\deg(f')\leq -\ee_n$. Then $\ol{\mu}^{(n)}+f'\cdot \ol{\mu}'^{(n)}$ is a minimal realisation of $s$.
\end{proposition}
\begin{proof} For  a proof that ${\mu}^{(n)}+f'\cdot {\mu}'^{(n)}\in\Min(s)$ using the index function, see \cite[Proposition 4.9]{N09d}. The proof that we have a minimal realisation now follows from Lemma \ref{f2}.
\end{proof}
\begin{theorem}\label{MRthm} 
(\cite[Theorem 4.16]{N95b})  Let $n\geq 1$, $D$ be a field and $s\in D^n$. If $\ol{f}\in\MR(s)$, then $\ol{f}=\ol{\mu}^{(n)}+f'\cdot \ol{\mu}'^{(n)}$ for some $f'\in D[x]$, where $f'=0$ or $\deg(f')\leq -\ee_n$. In particular, if $\ee_n>0$, then $\ol{\mu}^{(n)}$ is unique.
\end{theorem}

\subsection{The Iterative Algorithm}

As for the minimal polynomial algorithm, we write  $\ol{\mu}^{'(j)}$ for $\ol{\mu}^{(j')}$ and then suppress $j$. 
 
\begin{algorithm}   \label{MRa}\  
\begin{tabbing}

Input: \ \ \=$n\geq 1$ and $s=(s_1,\ldots,s_{n})\in D^n$.\\

Output: \>$\ol{\mu}\in\MR(s)$ and its MR identity.\end{tabbing}
\{$\Delta':=1$;\ $\ol{\mu}'  :=(\varepsilon,-1)$;\ $e := 1$;\ $\ol{\mu}  :=  (1,0)$;

\begin{tabbing}
{\bf for} \= $j = 1$ {\tt TO }$n$\\
    \> \{$\Delta    :=  \sum_{k=0}^{\frac{j-e}{2}} \mu_k \  s_{k+\frac{j+e}{2}};$ \\
   \> {\tt IF} $\Delta  \neq  0$ {\tt THEN }\{{\tt IF } $e\leq 0$  \=   						{\tt THEN} \=$\ol{\mu} :=  \Delta'\cdot \ol{\mu}-\Delta\cdot  x^{-e}\ \ol{\mu}'$\\ \\

  \>                 \>    {\tt ELSE} \{\=$\ol{t}:=\ol{\mu}$;\\
  \>                 \>           \> $\ol{\mu}:=   \Delta'\cdot x^e\ \ol{\mu}-\Delta\cdot \ol{\mu}'$;\\
  \>                 \>            \> $\ol{\mu}' :=\ol{t};\,\Delta':=\Delta$;\\
  \>                 \>            \> $e := -e$\}\}\\
  \> $e  := e+1$;\}\\
{\bf return} $(\ol{\mu})$\}
\end{tabbing}
\end{algorithm}
\begin{table}\label{shortex}
\caption{Algorithm \ref{MRa} with input 
$(0,1,1,0, 0,1,0,1)\in\F_2^8$.}
\begin{center}
\begin{tabular}{|c|r|r|l|l|l|l|l|l|}\hline
$j$ & $\Delta_j$  & $e_{j-1}$      			&$\mu^{(j)}$ &$\mu'^{(j)}$& $\mu_2^{(j)}$  & $\mu_2'^{(j)}$\\\hline\hline
$0$   &$1$        &$ $  & $1$ & $0$ &$0$ &$1$ \\\hline
$1$   &$0$  &$1$  & $1$ &$0$&$0$  &$1$ \\\hline
$2$   & $1$ &$2$ & $x^2$ & $1$ &$1$ &$0$  \\\hline
$3$   & $1$ &$-1$  & $x^2+x$ & $1$&$1$  &$0$  \\\hline
$4$   &$1$  &$0$  & $x^2+x+1$ & $1$&$1$ &$0$  \\\hline
$5$   & $1$ &$1$  & $x^3+x^2+x+1$  & $x^2+x+1$&$x$ &$1$ \\\hline
$6$   &$0$  &$0$  & $x^3+x^2+x+1$ &  $x^2+x+1$&$x$ &$1$  \\\hline
$7$   & $1$ & $1$ & $x^4+x^3+1$ & $x^3+x^2+x+1$&$x^2+1$ &$x$ \\\hline
$8$   &$1$  & $0$ & $x^4+x^2+x$ &  $x^3+x^2+x+1$&$x^2+x+1$&$x$\\\hline
\end{tabular}
\end{center}
\end{table}

\begin{example} \label{mrex} (Cf. \cite[Table I]{Salagean}) 
Table \ref{shortex} gives the values of $\ol{\mu}^{(j)}$ and $\ol{\mu}'^{(j)}$. The reader may check that $\tilde{\mu}'^{(j)}\cdot\ol{\mu}^{(j)}=1$ for $1\leq j\leq 8$ as per Theorem \ref{numu}.
\end{example}
Over a field, we can produce a monic $\mu$ by dividing by $\Delta'$.  We remark that  $\mu$ and $\mu_2$ can even be computed in parallel.
\section{Minimal Polynomial Identities} 

\subsection{An Identity for $\mu^{(n)}$ and $\mu'^{(n)}$} \label{gcds}
Our first set of identities was suggested by the following result.
\begin{proposition} \label{gcd} Let $D$ be a factorial domain and $s\in D^n$.
For $1\leq j\leq n$, let $e=\ee_{j-1}$. Then 

 (i) $\gcd(\mu^{(j)},\mu'^{(j)})=1$
 
 (ii) if  $\Delta_j\neq 0$ and $p_j\geq 0$ is the highest power of $x$ dividing $\mu^{(j)}$, then 
$$\gcd(\mu^{(j-1)},\mu^{(j)})=\left\{\begin{array}{ll}
 x^{\min\{p_{j-1},-e\}} & \mbox{ if } e\leq 0\\
 1 & \mbox{ if } e>0.
\end{array}
\right.
$$
(iii) $\gcd(\mu^{(j)\ast},\mu^{(j-1)\ast})=1$.
\end{proposition}
\begin{proof} We prove this by induction on $j$, the case $j=0$ being trivial. Suppose inductively that $\gcd(\mu^{(j-1)},\mu'^{(j-1)})=1$. Then

$$\mu^{(j)}= \left\{\begin{array}{ll}

\Delta_j\cdot \mu^{(j-1)}- \Delta_j\cdot  x^{-e}  \mu'^{(j-1)}
	& \mbox{ if } e\leq 0\\
\Delta'_{j}  \cdot x^{+e}\mu^{(j-1)}- \Delta_j\cdot 
 \mu'^{(j-1)} &\mbox{ otherwise}
\end{array}
\right.
$$
and
$$\mu'^{(j)}= \left\{\begin{array}{ll}

\mu'^{(j-1)}
	& \mbox{ if } e\leq 0\\
 \mu^{(j-1)} &\mbox{ if } e>0.
\end{array}
\right.
$$
Thus if $e\leq 0$, then $\gcd(\mu^{(j)},\mu'^{(j)})=\gcd(\mu^{(j)},\mu'^{(j-1)})=\gcd(\mu^{(j-1)},\mu'^{(j-1)})=1$ by the inductive hypothesis. If $e>0$, then $$\gcd(\mu^{(j)},\mu'^{(j)})=\gcd(\mu^{(j)},\mu^{(j-1)})=\gcd(\mu'^{(j-1)},\mu^{(j-1)})=1$$ by the inductive hypothesis.

(ii) Suppose that $e\leq 0$. Then $\mu^{(j)}=
\Delta'_{j}\cdot \mu^{(j-1)}- \Delta_j\cdot  x^{-e}\mu'^{(j-1)}$
and so $\gcd(\mu^{(j-1)},\mu^{(j)})=\gcd(\mu^{(j-1)},x^{-e}\mu'^{(j-1)})= x^{\min\{p_{j-1},-e\}}$ by  Part (i). For
$e>0$, we have $\mu^{(j)}=
\Delta'_j\cdot x^{e}\mu^{(j-1)}- \Delta_j\cdot \mu'^{(j-1)}$ and so $$\gcd(\mu^{(j-1)},\mu^{(j)})=\gcd(\mu^{(j-1)},\mu'^{(j-1)})=1.$$

(iii) Easy consequence of (ii). 
\end{proof}
\bc \label{relprime} Let $D$ be factorial and $s\in D^n$. The successive 'connection' polynomials of $s$ (obtained by the Berlekamp-Massey algorithm for $D$ of \cite[Algorithm 5.4]{N10a}) are  relatively prime. 
\ec
\begin{proof} For $1\leq j\leq n$, the successive connection polynomials of $s$ are $\mu^{(j)\ \ast}$ and $\mu^{(j-1)\ \ast}$ by \cite[Corollary 4.10]{N10a}. The result now follows from Proposition \ref{gcd}(iii).
\end{proof}
We will see that the following inductively defined scalar plays a key role: it turns out to be the right-hand side of all three identities.
\begin{definition}\label{nabla} Let $s\in D^n$ and for $1\leq j\leq n$, let $\Delta_{j}$,  and $e=\ee_{j-1}$  be as in Theorem \ref{bimrt}. Put $\nabla_0=1$ and for $1\leq j\leq n$, 
define $\nabla_{j}=\nabla_{j}(s)$ by 

(i)  $\nabla_{j}=\nabla_{j-1}$ if $\Delta_{j}=0$

(ii) 
$$\nabla_{j}= \left\{\begin{array}{ll}
\Delta'_{j}\cdot\nabla_{j-1} & \mbox{if } e\leq 0,\ \Delta_j\neq 0\\\\
\Delta_j\cdot\nabla_{j-1}
   & \mbox{if } e> 0,\ \Delta_j\neq 0.
\end{array}
\right.
$$

\end{definition}
We have $\nabla_1=\Delta_1$ since $e_0=1$. It is easy to see that the $\nabla_{j}$ are always non-zero. 
We now construct an identity for $\mu^{(n)}$ and $\mu'^{(n)}$.
\begin{theorem}\label{fg} Let $D$ be a domain, $s\in D^n$ and for $1\leq j\leq n$, let $e=\ee_{j-1}$ be as in Theorem \ref{bit}. Put $\ol{f}^{(0)}=(1,0)$. Assume inductively that we have defined $\ol{f}^{(k)}$ for $0\leq k\leq j-1$.  
If $\Delta_j=\Delta_j(\mu^{(j-1)})=0$, put $\ol{f}^{(j)}=\ol{f}^{(j-1)}$. If $\Delta_j\neq 0$, put
$$\ol{f}^{(j)}= \left\{\begin{array}{ll}
(+f^{(j-1)},\Delta'_j\cdot f_2^{(j-1)}+{\Delta_j}\cdot   x^{-e} f^{(j-1)})
	& \mbox{ if } e\leq 0\\\\
(-f_2^{(j-1)},\Delta'_j\cdot x^{+e}f_2^{( j-1)}+ {\Delta_j}\cdot f^{( j-1)}) &\mbox{ if } e>0.
\end{array}
\right.
$$
Then for $0\leq j\leq n$, we have 
$\ol{f}^{(j)}\cdot({\mu}^{(j)},\mu'^{(j)})=\nabla_j$. \end{theorem}
\begin{proof} The case $n=0$ is a trivial verification. 
Suppose that  $\ol{f}^{(j)}\cdot(\mu^{(j)},\mu'^{(j)})=1$ for $0\leq j\leq n-1$. Put  $\mu=\mu^{(n-1)}$,  $e=\ee_{n-1}$,  $\Delta =\Delta_n$, $\mu'=\mu'^{(n-1)}$,  $\Delta'=\Delta'_n$ and $\ol{f}=\ol{f}^{(n-1)}$.  We can assume that $\Delta_n\neq 0$.  If $e\leq 0$, we have $\mu'^{(n)}=\mu'$ and 
\begin{eqnarray*}
f^{(n)}\mu^{(n)}
&=&f(\Delta'\cdot \mu- \Delta\cdot x^{-e}\mu')=\Delta'\cdot f\mu -\Delta\cdot  x^{-e}f\mu'\\
&=&\Delta'\cdot (\nabla_{n-1}-f_2\mu')- \Delta\cdot  x^{-e}f\mu'\\
&=&\nabla_n-(\Delta'\cdot f_2+\Delta\cdot  x^{-e}f)\mu'=\nabla_n-f_2^{( n)}\mu'^{(n)}.
\end{eqnarray*}
If $e>0$, we have  $\mu'^{(n)}=\mu$. Thus
\begin{eqnarray*}
f^{(n)}\mu^{( n)}&=& -f_2(\Delta'\cdot x^{e}\mu-
\Delta\cdot \mu')= -\Delta'\cdot x^{e}f_2\mu+\Delta\cdot f_2\mu'\\
&=& -\Delta'\cdot x^{e}f_2\mu+\Delta\cdot (\nabla_{n-1}-f\mu)
= -(\Delta'\cdot x^{e}f_2+ \Delta\cdot f)\mu+\nabla_n\\
&=& -f_2^{(n)}\mu'^{(n)}+\nabla_n.
\end{eqnarray*}
\end{proof}
When $D$ is a field, these amount to B\'ezout Identities for the $\mu^{(n)}$ and $\mu'^{(n)}$. 
We see again that $\gcd(\mu^{(n)},\mu'^{(n)})=1$ if $D$ is factorial. The general form of $f_2^{(n)}$ shows that we can compute it iteratively alongside $\mu^{(n)}$ in Algorithm
\ref{MRa}. 
Our first application follows from Proposition \ref{gcd} when $D$ is factorial.
\bc \label{nonzero} Let $D$ be any domain and $s\in D^n$. If ${\mu}^{(n)}_0=0$ then ${\mu}'^{(n)}_0\neq 0$ and similarly for  ${\mu}'^{(n)}_0$.
\ec
If needed, we can obtain $\deg(f_2^{(n)})$ from
$\deg(f^{(n)})+\LC_n=\deg(f_2^{(n)})+\LC'_n$.
 
\subsection{An Identity for ${\mu}^{(n-1)}$ and $\mu^{(n)}$}
We give an identity for ${\mu}^{(n-1)}$ and $\mu^{(n)}$ when $\Delta_n\neq 0$ and $e_{n-1}>0$.  When $\Delta_n\neq 0$, $\ee_{n-1}\leq 0$  and $D$ is factorial, there is a similar non-constructive identity (which does not explicitly involve $\nabla_n$). 

\begin{proposition}\label{nextid}
Let  $s\in D^n$, $\Delta_n\neq 0$ and $e=\ee_{n-1}$  be as in Theorem \ref{bimrt}.
If $e> 0$ and $\ol{f}^{(n-1)}\cdot(\mu^{(n-1)},\mu'^{(n-1)})=\nabla_{n-1}$ then 
$$(\Delta\cdot f^{(n-1)}+\Delta'\cdot x^{e}f_2^{(n-1)},-f_2^{(n-1)})\cdot({\mu}^{(n-1)},\mu^{(n)})=\nabla_n.$$
\end{proposition}
\begin{proof}  Put  $\mu=\mu^{(n-1)}$,  $e=\ee_{n-1}$,  $\Delta =\Delta_n$, $\mu'=\mu'^{(n-1)}$,  $\Delta'=\Delta'_n$ and $\ol{f}=\ol{f}^{(n-1)}$.  The identity is trivial if $n=1$. 
Suppose that  $\ol{f}^{(n-1)}\cdot(\mu^{(n-1)},\mu'^{(n-1)})=1$. Then
\begin{eqnarray*}
(\Delta\cdot f+\Delta'\cdot x^{e}f_2,-f_2)\cdot({\mu},\mu^{(n)})&=&\Delta \cdot f\mu+\Delta'\cdot x^ef_2\mu-f_2\mu^{(n)}\\
&=&\Delta\cdot (\nabla_{n-1}-f_2\mu')+\Delta' \cdot x^e f_2\mu-f_2\mu^{(n)}\\
&=&\nabla_{n}-\Delta\cdot  f_2\mu'+\Delta'\cdot  f_2x^e\mu-f_2\mu^{(n)}\\
&=&\nabla_{n}+f_2(\Delta'\cdot  x^e\mu-\Delta \cdot \mu')-f_2\mu^{(n)}=\nabla_{n}.
\end{eqnarray*}
\end{proof}
The case $e_{n-1}\leq 0$ is complicated by $\gcd({\mu}^{(n-1)},x^{-e})=x^m$ say. If $m=-e$, one may check that
$$(x^{-e}f^{(n-1)}+f_2^{(n-1)})\mu^{(n-1)}-f_2^{(n-1)}\mu^{(n-1)}=\nabla_{n-1}\cdot\gcd({\mu}^{(n)},{\mu}^{(n-1)}).$$

In general, we can apply the following result when $D$ is factorial. 
\bt\label{FN} (\cite[Corollary 2.7]{FN95}) Let $D$ be factorial and $a,b\in D[x]$, not both zero. There are $u\in D$ and $f,g\in D[x]$  such that $fa+gb=u\gcd(a,b)$.
\et
For a specific $s\in D^n$, we may even obtain $u$, $f$ and $g$ by applying algorithm XPRS of \cite{FN95} to ${\mu}^{(n-1)}$ and $\mu^{(n)}$. See also Theorem \ref{newBIt}. It would be useful to have an identity for general $D$ which involves $\nabla_n$ when $\Delta_n\neq 0$ and $e_{n-1}\leq 0$. 
\subsection{An Identity for $\mu^{(n)}$ and $\mu^{(n)}_2$}
 Recall that $\ol{\mu}^{(-1)}=(\varepsilon,-1)$, $\ol{\mu}^{(0)}=(1,0)$ and that if $\ol{f}=(f,f_2)$, then $\tilde{f}=(-f_2,f)$.
\begin{theorem} \label{numu} For $0\leq j\leq n$, we have
$\tilde{\mu}'^{({j})}\cdot\ol{\mu}^{(j)}=\nabla_{j}$. Hence if $
{\mu}^{({j})}_0=0$ then $({\mu}_2^{({j})})_0\neq 0$ and similarly for $({\mu}^{({j})}_2)_0$.
\end{theorem}
\begin{proof} For $j=0$, we have $(1,\varepsilon)\cdot(1,0)=1=\nabla_0$.  Suppose $1\leq j\leq n$ and the result is true for $k\leq j-1$. If $\Delta_j=0$, there is nothing to prove. Otherwise, let    $e=\ee_{j-1}$. If $e\leq 0$ then
\begin{eqnarray*}
\tilde{\mu}'^{(j)}\cdot\ol{\mu}^{(j)}&=&\tilde{\mu}'\cdot\ol{\mu}^{(j)}\\&=&
-\mu_2'
(\Delta'_{j} \cdot\mu-\Delta_j\cdot  x^{-e}\mu')
+\mu'(\Delta'_j \cdot\mu_2-\Delta_j\cdot  x^{-e}\mu_2')\\
&=&\Delta'_j \cdot( -\mu_2'\mu+\mu'\mu_2)
+\Delta_j\cdot  x^{-e}(\mu_2'\mu'-\mu'\mu_2')\\
&=&\Delta'_j\cdot \tilde{\mu}'\cdot\ol{\mu}=\Delta'_j\cdot\nabla_{j-1}=\nabla_{j} 
\end{eqnarray*}
and if $e>0$ then
\begin{eqnarray*}
\tilde{\mu}'^{(j)}\cdot\ol{\mu}^{(j)}&=&\tilde{\mu}\cdot\ol{\mu}^{(j)}\\
&=&-\mu_2
(\Delta'_j\cdot x^{e}\mu-\Delta_j\cdot  \mu')+
\mu(\Delta'_j \cdot x^{e}\mu_2-\Delta_j\cdot  \mu_2')\\
&=&\Delta'_j \cdot x^{e}(-\mu_2\mu+\mu\mu_2)
+\Delta_j\cdot  (-\mu_2'\mu+\mu'\mu_2)\\
&=&\Delta_j\cdot \tilde{\mu}'\cdot\ol{\mu}=\Delta_j\cdot \nabla_{j-1}=\nabla_{j}
\end{eqnarray*}
which completes the induction.
\end{proof}
When $D$ is a field, the identity for ${\mu}^{(n-1)}$ and $\mu^{(n)}$ of Theorem \ref{numu} amounts to a B\'ezout identity. 
The following result was proved differently in \cite[Corollary 3.24]{N95b}.
\begin{corollary} Let $D$ be a factorial domain and $s\in D^n$. Then $\gcd(\ol{\mu}^{(j)})\in D$ for $1\leq j\leq n$.
\end{corollary}

\begin{corollary} \label{anyMR} Let $D$ be a field and $s\in D^n$.  Then for $0\leq j\leq n$, $\tilde{\mu}'^{({j})}$ are {B\'{e}zout}  coefficients of any minimal realization of $s$.
\end{corollary}

\begin{proof} Let  $\ol{m}$ be a minimal realization of $s$. By Theorem \ref{MRthm}, $\ol{m}=\ol{\mu}^{(n)}+m'\ol{\mu}'^{(n)}=\ol{\mu}+m'\ol{\mu}'$ say, where $\deg(m')\leq -\ee_n$. Now $\tilde{\mu}'\cdot\ol{\mu}'=0$ gives $$\tilde{\mu}'\cdot\ol{m}=\tilde{\mu}'\cdot\ol{\mu}+m'\tilde{\mu}'\cdot\ol{\mu}'=\tilde{\mu}'\cdot\ol{\mu}=\nabla_{n}.$$
\end{proof}
\begin{proposition} \label{xea} Let $D$ be a field, $s\in D^n$ and B\'ezout coefficients $\ol{a}$ be computed for $\ol{\mu}^{(n)}$ by the extended Euclidean algorithm. Then
$\nabla_{n}\cdot\ol{a}=\tilde{\mu}'^{(n)}$. 
\end{proposition}
\begin{proof} Put $\ol{a}=(a,a_2)$, $\ol{\mu}=\ol{\mu}^{(n)}$, $\ol{\mu}'=\ol{\mu}'^{(n)}$ and $\nabla=\nabla_{n}$. We have $\ol{a}\cdot\ol{\mu}=\gcd(\ol{\mu})=1=\nabla^{-1}\cdot\tilde{\mu}'\cdot\ol{\mu}$.
Hence $(a+\nabla^{-1}\mu_2')\mu=-
(a_2-\nabla^{-1}\mu')\mu_2$. If  $a_2-\nabla^{-1}\mu_2=0$, we are done. Otherwise $\mu|(a_2-\nabla^{-1}\mu')$ since $\gcd(\ol{\mu})=1$. This is impossible since $\deg(a_2-\nabla^{-1}\mu')<\deg(\mu)$. Thus
$\ol{a}=\nabla^{-1}\tilde{\mu}'$.
\end{proof}
We can also apply Theorem \ref{numu} to linear recurring sequences.
\begin{proposition} \label{lrs} Let $D$ be a principal ideal domain and $s\in D^n$ be a linear recurring sequence with
 minimal polynomial $\mu$. There are explicit $\ol{f}\in D[x]^2$ and $\nabla\in D\setminus\{0\}$ such that $\ol{f}\cdot\ol{\mu}=\nabla$.
\end{proposition}
\begin{proof} Apply Theorem \ref{numu} to $s^{(2\deg(\mu))}$.
\end{proof}

\subsection{The Iterative Algorithm}

The following extension of Algorithm \ref{MRa} is immediately implied by Theorem \ref{numu}. 
\begin{algorithm}   \label{MRplusmult}\  
\begin{tabbing}
Input: \ \ \=$n\geq 1$ and $s=(s_1,\ldots,s_{n})\in D^n$.\\
Output: \>$\ol{\mu}\in\MR(s)$ and the MR identities  of $s$.\\\\
\{$\ol{\mu}'  :=(\varepsilon,-1)$;\ $\Delta':=1$;\  $e := 1;$\ $\ol{\mu}  :=  (1,0)$;\\\\

$\nabla:= 1$;\\
{\tt FOR} \= $j = 1$ {\tt TO }$n$\\
    \> \{$\Delta    :=  \sum_{k=0}^{\frac{j-e}{2}} \mu_k\   s_{k+\frac{j+e}{2}};$ \\
   \> {\tt IF} $\Delta  \neq  0$ {\tt THEN }\{{\tt IF} $e\leq 0$  \=   						{\tt THEN} \{\=$\ol{\mu} :=  \Delta'\cdot \ol{\mu}-\Delta\cdot  x^{-e}\ \ol{\mu}'$;\\ 
 \> \>  \> $\nabla:=\Delta'\cdot \nabla$\}\\
  \>                 \>    {\tt ELSE} \{\=$\ol{t}:=\ol{\mu}$;\\
  \>                 \>           \> $\ol{\mu}:=   \Delta'\cdot x^e\ \ol{\mu}-\Delta\cdot \ol{\mu}'$;\\
  \>                 \>            \> $\ol{\mu}' :=\ol{t};\,\Delta':=\Delta$;\\
  \>                 \>           \> $\nabla:=\Delta\cdot\nabla;$\\ 
  \>                 \>           \> $e := -e$\}\}\\ 
  \> $e  := e+1$;\}\\
{\tt RETURN} $(\ol{\mu},\tilde{\mu}',\nabla)$;\}
\end{tabbing}
\end{algorithm}
 We could also compute $\nabla_j$ at the end of iteration as $\ol{f}\cdot\ol{\mu}^{(j)}$, but this would increase the complexity of the algorithm.

 \begin{table}   \label{shortexM}
\caption{Theorem \ref{numu} with  $s=(1,0,1,1,0,1)\in \mathrm{GF}(2)^{6}$.}
\begin{center}
\begin{tabular}{|c|r|l|l|l|l|l|l|}\hline
$j$ & $e$  & $\Delta$  &$\ol{\mu}$    & $\ol{\mu}'$ &$\tilde{\mu}'$ \\\hline\hline
 $1$   &$1$ &$1$ &$(x,1)$ & $(1,0)$   &$(0,1)$\\\hline
$2$    & $0$  &$0$  &$(x,1)$& $(1,0)$  &$(0,1)$\\\hline
$3$   & $1$   &$1$ &$(x^2+1,x)$ &$(x,1)$ &$(1,x)$ \\\hline
$4$ &$0$ &$1$ &$(x^2+x+1,x+1)$ & $(x,1)$ &$(1,x)$\\\hline
$5$ & $1$ &$0$ &$(x^2+x+1,x+1)$ &$(x,1)$ &$(1,x)$\\\hline
$6$   & $2$ &$0$  &$(x^2+x+1,x+1)$& $(x,1)$ &$(1,x)$\\\hline
\end{tabular}
\end{center}
\end{table}

\begin{example} \label{Bezoutex} Let $s=(1,0,1,1,0,1)\in \mathrm{GF}(2)^{6}$.
Then $$\ol{f}^{(3)}\cdot\ol{\mu}^{(3)}=1\cdot (x^2+1)+x\cdot x=1$$
$$\ol{f}^{(j)}\cdot\ol{\mu}^{(j)}=1\cdot (x^2+x+1)+x\cdot(x+1)=1\mbox{ for }4\leq j\leq 6.$$
\end{example}

\section{Perfect Linear-Complexity Profile}
The following generalisation of a definition from \cite{Rueppel} to any domain $D$ makes sense by Theorem \ref{bit}.
 \subsection{Basics}
\bd A sequence $s\in D^n$ has a {\em perfect linear-complexity profile (PLCP)} if   $\LC_j = \lfloor \frac{j+1}{2}\rfloor$ for $1\leq j\leq n$.
\ed
It is easy to see that for $D=\F_2$, the sequences of length 1 to 4 with a PLCP are $(1)$, $(1,s_2)$, $(1,1,0), (1,0,1)$ and $(1,1,0,s_4), (1,0,1,s_4)$. See Table \ref{plcp}.2 for the corresponding $\ol{\mu}^{(j)}$, where $\Delta_j=\Delta(\mu^{(j-1)})$.
Recall that for any sequence, $\mu^{(0)}=1$ and $e_0=1$.

\bp \label{basictfae}\tfae

(i) $s$ has a PLCP

(ii) for $1\leq j\leq n$
$$e_j= \left\{\begin{array}{rl}

         1& \mbox{if } j \mbox{ is even}\\
	0 &\mbox{otherwise,}
 \end{array}
\right. $$

(iii) $\Delta_{j}\neq 0$ for all odd $j$, $1\leq j\leq n$

(iv) $\ol{\mu}^{(1)}=(x+\varepsilon,1)$ and for $2\leq j\leq n$
$$\ol{\mu}^{(j)}= \left\{\begin{array}{ll}
        \Delta_{j-1}\cdot \ol{\mu}^{(j-1)}-\Delta_j\cdot\ol{\mu}^{(j-2)}&\mbox{ if } j \mbox{ is even}\\\\
\Delta_{j-2}\cdot x\ol{\mu}^{(j-1)}-\Delta_{j}\cdot\ol{\mu}^{(j-3)}& \mbox{ otherwise.}
 \end{array}
\right. $$
\ep
\bpr
(i) $\Leftrightarrow$ (ii): Easy consequence of the definitions.

(i) $\Rightarrow$ (iii): If $j\leq n+1$ is odd then $\Delta_j\neq 0$, for otherwise $\frac{j-1}{2}+1=\frac{j+1}{2}=\LC_{j}=\LC_{j-1}=\frac{j-1}{2}$. 

(iii) $\Rightarrow$ (i): Let  $\Delta_{j}\neq 0$ for all odd $j$, $1\leq j\leq n+1$. Then $s_1\neq 0$, $\LC_1=1$ and $\ee_1=0$. If $\Delta_2=0$, then $\LC_2=\LC_1=1$, otherwise $\LC_2=\max\{\ee_1,0\}+1=1$, so that $\LC_2$ is as required. Suppose that $j\leq n$ is odd and $\LC_{k}=\lfloor\frac{k+1}{2}\rfloor$ for all $k$, $1\leq k\leq j-1$. We have $\LC_j=j-\LC_{j-1}=j-\frac{j-1}{2}=\lfloor\frac{j+1}{2}\rfloor$. If $j=n+1$, we are done. Otherwise,  if $\Delta_{j+1}=0$, we have $\LC_{j+1}=\LC_j=\lfloor\frac{j+1}{2}\rfloor=\lfloor\frac{j+2}{2}\rfloor$, whereas if $\Delta_{j+1}\neq 0$,
$\LC_{j+1}=j+1-\LC_j=j+1-\lfloor\frac{j+1}{2}\rfloor=\lfloor\frac{j+2}{2}\rfloor$.
(ii) $\Ra$ (iv):  We have 
$$\ol{\mu}^{(j)}=  \Delta'_{j}\cdot x^{\max\{\ee_{j-1},0\}}\ \ol{\mu}^{(j-1)}-\Delta_j\cdot x^{\max\{-\ee_{j-1},0\}}\ \ol{\mu}'^{(j-1)}$$
 where $\Delta_j$ may be zero and $\ee_{j-1}$ is as usual.
If $j$ is even then $\Delta_{j-1}\neq 0$ and $\ee_{j-2}=1$, so $i=j-2$; if $j$ is odd then $e_{j-2}=0$ and $e_{j-3}=1$ so $(j-2)'=j-3$. Inserting  $\ee_{j-1}$ now gives the formulae for $\ol{\mu}$. (Note that $(j-1)'+1$ is always odd, so that $\Delta'_j\neq 0$.)

(iv) $\Ra$ (ii): We have $\LC_1=1$, $\LC_j=\LC_{j-1}$ if $j$ is even and $\LC_j=\LC_{j-1}+1$ if $i$ is odd. Thus if $j$ is odd,
$\ee_j=j+1-2\LC_j=j+1-2(\LC_{j-1}+1)=j+1-2\LC_{j-1}=\ee_{j-1}+1$. Applying this inductively gives (ii).
\epr
\brs
(i) We applied Theorem \ref{bimrt} to $s\in D^n$, not to $(\Delta_1,\ldots,\Delta_{n})\in D^n$ as stated in \cite{Wangpdf}. 

(ii) If $s$ is an infinite linear recurring sequence [i.e. $\ul{s}\in D[[x^{-1}]]$ is a rational function $g/f$ with $\deg(g)<\deg(f)$] then $\lim_{n\ra\infty}\ee_{n}=\infty$ whereas if $s$ has a PLCP, trivially $\ul{\lim}_{n\ra\infty}\ee_{n}=0$, $\ol{\lim}_{n\ra\infty}\ee_{n}=1$ but $\lim_{n\ra\infty}\ee_{n}$ does not exist. So linear recurring sequences never have PLCP.

(iii) For  an infinite sequence $s$, Proposition \ref{basictfae} is consistent with \cite{Nied86b}. Let  $$\mathrm{K}(\ul{s})=\sup_{n\geq 1}\{\deg(A_n)\}$$  as in \cite{Nied86b}. Then if $\ul{s}$ is not a rational function and has partial quotients $A_i$, Theorem 1 of \cite{Nied86b}
implies that for all $n$, $1-\mathrm{K}(\ul{s})\leq \ee_n\leq \mathrm{K}(\ul{s})$ and 
$s$ has  PLCP if and only if $\ul{s}$ is not a rational function and $\mathrm{K}(\ul{s})=1$ \cite[Theorem 3]{Nied86b}. 
\ers
\bp\label{count} If $D=\F_q$, then $|\{s\in D^n:\ s \mbox{ has  PLCP}\}|= (q-1)^{\lceil \frac{n}{2}\rceil}q^{\lfloor \frac{n}{2}\rfloor}$, and if $|D|=\infty$ then $\{s\in D^n:\ s \mbox{ has  PLCP}\}$ is infinite.
\ep
 \bpr A sequence $s$ determines $(\Delta_1,\ldots,\Delta_{n})\in D^n$ uniquely, and conversely. Thus the result follows from Proposition \ref{basictfae}(iii).\epr
 \begin{table}   \label{plcp}
\caption{$\ol{\mu}^{(j)}$, ${\sigma}^{(j)}\in \F_2[x]$ for $0\leq i\leq 4$.}
\begin{center}
\begin{tabular}{|r|l|l|}\hline
$j$ & $\ol{\mu}^{(j)}$&$\sigma^{(j)}$\\\hline\hline
$0$ & $(1,0)$&$1$\\\hline
$1$ & $(x,1)$&$x+1$\\\hline
$2$ & $(x+\Delta_1,1)$&$(\Delta_1+1)x+1$\\\hline
$3$ & $(x^2+\Delta_1 x+1,x)$&$(\Delta_1+1)x^3+x+1$\\\hline
$4$ & $(x^2+(\Delta_1+\Delta_3)x+1,x+\Delta_3)$&$(\Delta_1+1)x^3+(\Delta_1\Delta_3+1)x+1$\\\hline
\end{tabular}
\end{center}
\end{table}
It follows that without loss of generality, we may assume that $n$ is odd.
It is easy to see that for odd $n$, $\ol{\mu}^{(n)}+c\cdot \ol{\mu}'^{(n)}=\ol{\mu}^{(n)}+c\cdot \ol{\mu}^{(n-1)}\in \mathrm{MR}(s)$ for any $c\in D$.

\bc Let $D$ be a field and $\ol{\mu}^{(j)}$  and $\Delta_j=\Delta_j(\mu^{(j-1)})$ be as in Theorem \ref{MRthm} for $1\leq j\leq n$. If $s$ has PLCP, then 
$$\mathrm{MR}(s^{(j)})= \left\{\begin{array}{ll}
\{\ol{\mu}^{(j-1)}-\frac{\Delta_j}{\Delta_{j-1}}\cdot\ol{\mu}^{(j-2)}\} 
            &  \mbox{ if } j \mbox{ is even}\\\\
\{(x+c)\cdot\ol{\mu}^{(j-1)}-\frac{\Delta_j}{\Delta_{j-2}}\cdot\ol{\mu}^{(j-3)}:\ c\in D\}&\mbox{ otherwise.}
 \end{array}
\right. $$
\ec

\bpr If $j$ is even, then $s^{(j)}$ has a unique monic minimal realisation by Theorem \ref{MRthm}. Similarly if $j$ is odd,  any monic minimal realisation of $s^{(j)}$ is 
$$\ol{\mu}^{(j)}+c\cdot\ol{\mu}'^{(j)}
=(x\ol{\mu}^{(j-1)}-\frac{\Delta_j}{\Delta_{j-2}}\cdot\ol{\mu}^{(j-3)})+c\cdot\ol{\mu}^{(j-1)}=(x+c)\cdot\ol{\mu}^{(j-1)}-\frac{\Delta_j}{\Delta_{j-2}}\cdot\ol{\mu}^{(j-3)}$$ for some $c\in D$.
\epr
\bc \label{plcpcoeffs}
If $s$ has PLCP then for $2\leq j\leq n$, then $\ol{g}^{(j)}\cdot\ol{\mu}^{(j)}=\nabla_{j}$ where
 $$\ol{g}^{(j)}=\left\{\begin{array}{ll}
  (-\mu_2^{(j-2)},\mu^{(j-2)}) & \mbox{if } j 	      \mbox{ is even}\\\\
(-\mu_2^{(j-1)},\mu^{(j-1)}) &\mbox{otherwise}
\end{array}
\right.
$$
and 
$$\nabla_{j}=\left\{\begin{array}{ll}
 (\Delta_{j-1}\cdot\Delta_{j-3}\cdots \Delta_1)^2  & \mbox{if } j 	      \mbox{ is even}\\\\
 \Delta_j\cdot (\Delta_{j-2}\cdot\Delta_{j-4}\cdots \Delta_1)^2 &\mbox{otherwise.}
\end{array}
\right.
$$

\ec
\bpr We have $g^{(j)}=\tilde{\mu}'^{(j)}$ and $j'=j-2$ if $j$ is even and $j'=j-1$ if $j$ is odd. To verify the values for $\nabla_{j}$, we have $e_j=1$ if $j$ is even and $e_j=0$ if $j$ is odd. Thus if $j\geq 2$ is even, $\nabla_{j}=\Delta_{j-1}\cdot\nabla_{j-1}$ and if $j$ is odd,  $\nabla_{j}=\Delta_{j}\cdot\nabla_{j-1}$, which easily yields the stated formulae.
\epr
With $D=\F_2$ for example, $(0,1)\cdot(x+\Delta_1,1)=1$ and 
 $(1,x+\Delta_1)\cdot(x^2+\Delta_1 x+1,x)
=(x^2+\Delta_1 x+1) +x(x+\Delta_1 )=1$.

\subsection{Binary Sequences}
Here we give a simpler proof of a theorem of  Wang and Massey \cite{Wangpdf} on binary sequences via B\'ezout coefficients. Thus $D=\F_2$ throughout this subsection.
Let us call $s\in D^n$ {\em stable} if $s_1=1$ and for even $j$, $2\leq j\leq n$, $s_{j+1}=s_j+s_{\frac{j}{2}}$.
It is convenient to use two auxiliary functions
$$\ul{t}=\ul{t}(s)=\ul{s}^2+(x+1)\ul{s}+1$$ and
$$\sigma^{(j)}=\nu^{(j)2}+(x+1)\nu^{(j)}\mu^{(j)}+\mu^{(j)2}$$ 
where $\ol{\mu}^{(j)}=(\mu^{(j)},\nu^{(j)})$ is as in Proposition
\ref{basictfae}(iv); $\ul{t}$ is also used in the continued fraction treatment of this theorem in  \cite{Nied87} (for infinite sequences).  First we have the following simple consequence of the definitions.
\bp $\mu^{(n)2}\ \ul{t}=\sigma^{(n)}+F$ where $\vv(F)\leq 0$ if $n$ is even and $\vv(F)\leq 1$ if $n$ is odd.
\ep
\bl \label{ult} 
Let $n$ be odd. \tfae

(i) $s$ is stable

(ii) $t_j=0$ for $j$ even, $0\leq j\leq n$.
\el
\bpr We have $t_0=s_1+1$ and for all even $j$, $2\leq j\leq {n-1}$, we have ${t}_j=s_j+s_{j+1}+s_{\frac{j}{2}}$.
\epr
The innocuous-looking part (ii) of the following lemma is essential to our proof.
\bl (Cf. \cite[Lemma 1]{Wangpdf})\label{sigma} If $s$ has PLCP then  (i) for $2\leq j\leq n$,
 $$\sigma^{(j)}=\left\{\begin{array}{ll}
\sigma^{(j-1)}+\Delta_i\cdot \sigma^{(j-2)}
+\Delta_j\cdot (x+1)& \mbox{if } j 	      \mbox{ is even}\\\\
x^2\sigma^{(j-1)}+\sigma^{(j-3)}+x(x+1) &\mbox{otherwise}
\end{array}
\right.
$$ 
(ii) for $0\leq j\leq n$, $\sigma^{(j)}_0=1$ and $\sigma^{(j)}_2=0$.
\el
\bpr Put $\nu=\mu_2$. 
If $j$ is odd, $\ol{\mu}^{(j)}=x\ol{\mu}^{(j-1)}+\ol{\mu}^{(j-3)}$, so
\begin{eqnarray*}\sigma^{(i)}&=&[x^2\nu^{(j-1)2}+\nu^{(j-3)2}]+(x+1)[x\mu^{(j-1)}+\mu^{(j-3)}][x\nu^{(j-1)}+\nu^{(j-3)}]\\&&\ \ +[x^2\mu^{(j-1)2}+\mu^{(j-3)2}]\\
&=&x^2\sigma^{(j-1)}+[\nu^{(j-3)2}+\mu^{(j-3)2}]+(x+1)[x\mu^{(j-1)}\nu^{(j-3)}+x\mu^{(j-3)}\nu^{(j-1)}+\mu^{(j-3)}\nu^{(j-3)}]\\
&=&x^2\sigma^{(j-1)}+\sigma^{(j-3)}+x(x+1)[\nu^{(j-1)}\mu^{(j-3)}+\mu^{(j-1)}\nu^{(j-3)}].
\end{eqnarray*}
We have $\ol{\mu}^{(j-1)}=\ol{\mu}^{(j-2)}+\Delta_{j-1}\cdot\ol{\mu}^{(j-3)}$, so 
\begin{eqnarray*}\nu^{(j-1)}\mu^{(j-3)}+\mu^{(j-1)}\nu^{(j-3)}&=&[\nu^{(j-2)}+\Delta_{j-1}\cdot\nu^{(j-3)}]\mu^{(j-3)}
+[\mu^{(j-2)}+\Delta_{j-1}\cdot\mu^{(j-3)}]\nu^{(j-3)}\\
&=&\nu^{(j-2)}\mu^{(j-3)}+\mu^{(j-2)}\nu^{(j-3)}
+2\Delta_{j-1}\cdot\nu^{(j-3)}\mu^{(j-3)}\\
&=&\nu^{(j-2)}\mu^{(j-3)}+\mu^{(j-2)}\nu^{(j-3)}=1
\end{eqnarray*}
by Corollary \ref{plcpcoeffs} and so for $i$ odd,
$\sigma^{(i)}=x^2\sigma^{(j-1)}+\sigma^{(j-3)}+x(x+1)$.
If $i$ is even, $\ol{\mu}^{(i)}=\ol{\mu}^{(j-1)}+\Delta_i\cdot\ol{\mu}^{(j-2)}$ and  Corollary \ref{plcpcoeffs} yields 
\begin{eqnarray*}\sigma^{(i)}&=&[\nu^{(j-1)2}+\Delta_i\cdot\nu^{(j-2)2}]+(x+1)[\nu^{(j-1)}+\Delta_i\cdot\nu^{(j-2)}][\mu^{(j-1)}+\Delta_i\cdot\mu^{(j-2)}]\\&&\ \ +[\mu^{(j-1)2}+\Delta_i\cdot\mu^{(j-2)2}]\\
&=&\sigma^{(j-1)}+\Delta_i\cdot\sigma^{(j-2)}+\Delta_i\cdot(x+1)[\nu^{(j-2)}\mu^{(j-1)}+\mu^{(j-2)}\nu^{(j-1)}]\\
&=&\sigma^{(j-1)}+\Delta_i\cdot\sigma^{(j-2)}+\Delta_i\cdot(x+1)
\end{eqnarray*}
by Corollary \ref{plcpcoeffs}. We omit the simple inductive proof of Part (ii).
\epr
The reader may also check that $\deg(\sigma^{(j)})=
 j-1$ if $j$ is even and $j$ if $\deg(\sigma^{(j)})=j$ is odd, but we will not need this fact.

\bt \label{stable} Let  $s\in D^n$. Then $s$ has PLCP if and only if $s$ is stable.
\et
\bpr It suffices to show that if $n$ is odd and $s$ has PLCP then ${t}_{j}=0$ for even $j$, $0\leq j\leq n-1$.
Put $\mu=\mu^{(j)}$. Inductively, we have
$$\mu^2\cdot(t_1x^{-1}+\cdots+ t_{n-2}x^{2-n}+t_{n-1}x^{1-n}+t_nx^{-n})=\mu^2\cdot\ul{t}=\sigma+F$$
where $\vv(F)\leq 1$. We want to show that $t_{n-1}=0$. The coefficient of $x^2$ in the left-hand side is $\mu_{\frac{n+1}{2}}\cdot t_{n-1}$ and $\sigma^{(n)}_2=0$ by Lemma \ref{sigma}(ii). Since $\mu_{\frac{n+1}{2}}$ is the leading coefficient of $\mu$, we must have $t_{n-1}=0$ . Conversely, as in noted in \cite{Wangpdf}, there are clearly $\lceil \frac{n}{2}\rceil$ stable sequences in $D^n$, so the result follows from Proposition \ref{count}.
\epr
\section{An Algorithm for the Polynomial B\'{e}zout Identity.}

\begin{theorem} \label{newBIt} (Cf. \cite[Corollary 2.7]{FN95}) Let $D$ be a principal ideal domain, $\ol{u}\in D[x]^2$, with $u$ monic and $0\leq \deg(u_2)\leq \deg(u)$. There are explicit $\ol{f}\in D[x]^2$ and $\nabla\in D\setminus\{0\}$ such that  $\ol{f}\cdot\ol{u}=\nabla\cdot\gcd(\ol{u})$.
\end{theorem}
\begin{proof} If $\deg(u_2)=\deg(u)=d$, let $\ell,\ \ell_2$ be the leading coefficients of $u$ and $u_2$ respectively. We replace $u_2$ by $u_2'=\ell_2\cdot u-\ell\cdot u_2$. If $\ol{f}\cdot (u, u_2')=\nabla\cdot\gcd(\ol{u})$, one checks that $(f+\ell_2f_2,-\ell f_2)\cdot \ol{u}=\nabla\cdot\gcd(\ol{u})$. So we can assume that $\deg(u_2)<\deg(u)$.
Let  $t$ be the linear recurring sequence defined by $\ul{t}=\frac{v}{u}\in D[[x^{-1}]]$, $s=t^{(2d)}$  and $\ol{\mu}$ be the unique MR of $s$, with $\ol{f},\nabla$ as in Proposition \ref{lrs}, so that $\ol{f}\ol{\mu}=\nabla$.

We know that $u\in\Ann(s)=\mu D[x]$ and so $u=w\mu$ for some $w\in D[x]$. We show that $w=\gcd(\ol{u})$. For   $u_2=u t=(w\mu) t=w(\mu t)=w\mu_2$ i.e.
$ \ol{u}=w\cdot \ol{\mu}$, so $\gcd(\ol{u})|w$. We also  know that there exists $\ol{a}\in D[x]$ such that $\ol{a}\cdot\ol{u}=\gcd(\ol{u})$. Then $\gcd(\ol{u})= \ol{a}\cdot w\ol{\mu}
=w \ol{a}\cdot \ol{\mu}$, so $w|\gcd(\ol{u})$.
Finally, $\nabla\cdot\gcd(\ol{u})=\nabla \cdot w=w\cdot(\ol{f}\cdot \ol{\mu})=\ol{f}\cdot (w\ol{\mu})=\ol{f}\ol{u}$.
\end{proof}

As in Proposition \ref{xea}, a degree argument shows that $\nabla^{-1}\ol{f}$ agrees with the coefficients found by the extended Euclidean algorithm. Thus Algorithm \ref{newBezouta} will be widely applicable, not just to $\F[x]$, $\F$ a field, but also to $\F[x,y]$ for example, as $\F[x]$ is a principal ideal domain.
\begin{algorithm} \label{newBezouta} (Cf. \cite[Algorithm 4.2]{FN95})
\begin{tabbing}
Input:\ \ \ \= P.I.D. $D$, $\ol{u}\in D[x]^2$, with $u$ monic and $0\leq \deg(u_2)<\deg(u)=d$.\\
Output: \> $\ol{f}\in D[x]^2$ such that  $\ol{f}\cdot\ol{u}=\nabla\cdot\gcd(\ol{u})$.
\end{tabbing}
1. Compute $s=t^{(2d)}$  by $2d$ subtractions of $u$ from $u_2$ in $D[x^{-1},x]$, where $\ul{t}=\frac{u_2}{u}$. 

\noindent 2. Apply Algorithm \ref{MRplusmult} to $s$, giving $\ol{f}\in D[x]^2$  and $\nabla\in D$. 

\noindent 3. Compute  $\ol{f}\cdot\ol{u}$, which is $\nabla\cdot\gcd(\ol{u})$.
\end{algorithm}
Note that step 2 requires at most $5d^2$ multiplications (and at most $3d^2$ multiplications
if $D$ is a field).
\begin{example}  Let $D= \mathrm{GF}(2)$,  $\ol{u}=(x^3+1,x^2+1)$ with $\gcd(\ol{u})=x+1$. Then $2d$ subtractions    of $u$ from $u_2$ in  $D[x^{-1},x]$ gives the generating function $x^{-1}+x^{-3}+x^{-4}+x^{-6}$ of $s$ i.e. 
$s=(1,0,1,1,0,1)$ as in Example \ref{Bezoutex}. We get
$1\cdot(x^3+1)+x\cdot(x^2+1)=x+1$.

\end{example}

 \section{Annihilators Which Do Not Vanish At Zero}

First we establish a lower bound lemma which was stated without proof in \cite{Salagean}. The resulting easy corollary and algorithm do not require any characterisations of minimal polynomials. Theorem \ref{nonzero} (or Proposition \ref{gcd} if $D$ is factorial) allows us to remove a test on $\mu'$. We also generalise Proposition \ref{MRprop} and Theorem \ref{numu} to give another construction of these annihilators (by extending the original sequence by one term) and include minimal polynomial identities.  In the final subsection, we prove a characterisation which was also stated without proof in \cite{Salagean}.
\subsection{The Lower Bound Lemma}
First of all, $s\in D^n$ and we can assume that $s$ is not the all-zero sequence, so that $\LC(s)\geq 1$. It is convenient to write $\Ann^\bullet(s)=\{f\in\Ann(s):\  f_0\neq 0\}$. We seek an $f\in\Ann^\bullet(s)$  such that $\deg(f)=\min\{\deg(g):\ g\in\Ann^\bullet(s)\}$.
To simplify the notation, we write  $\ol{\mu}=\ol{\mu}^{(n)}$,  $e=\ee_n>0$ and $\LC=\LC_n$, each obtained as in Theorem \ref{bimrt}. We also write  $\LC'=\LC_{n'}$,  $e'=\ee_{n'}$,    $\ol{\mu}'=\ol{\mu}^{(n')}$,
 $$
\MR^\bullet(s)=\{(f,f_2):\ f\in\Ann^\bullet(s), \deg(f)\mbox{ is minimal}\},
$$
$\LC^\bullet=\LC^\bullet(s)=\min\{\deg(f):\ f\in \Ann^\bullet(s)\}$ and for $n'\geq 1$, $s'=(s_1,\ldots,s_{n'})$.

If $\mu_0\neq 0$, clearly $\ol{\mu}\in \MR^\bullet(s)$  and evidently $\MR^\bullet(s_1,\ldots,s_{n+1})
\subseteq\MR^\bullet(s_1,\ldots,s_{n})$. 
The case $e\leq 0$ is an easy consequence of Proposition \ref{MRprop}, but the case $e>0$ requires some preparation. We begin with the following easy result.
\begin{proposition}\label{shift} 
Let  $f\in\Ann(s)$ and  $d=\deg(f)$.

(i) If $1\leq k\leq d$ and $x^k|f$ then $f/x^k\in \Ann(s_{k+1},\ldots,s_n)$.

(ii) If $2\leq k\leq n$ and $g\in \Ann(s_{k},\ldots,s_n)$, then  for any $t\in  D^{k-1}$, $x^{k-1}g\in \Ann(t_1,\ldots,t_{k-1},s_k,\ldots,s_n)$.
\end{proposition}
\bpr (i) Let $f'=f/x$ and $s'=s_2x^{-1}+\cdots +s_nx^{1-n}$. For $d-n=(d-1)-(n-1)\leq j\leq -1$, 
$(f's')_j=(f\ul{s})_j-s_1f_{j+1}=0$ since $f\in\Ann(s)$  and $j+1\leq 0$. Now induct on $k$. 

(ii) If  $g\in\Ann(s)$ then
for any $t\in  D$, $xg\in\Ann(t,s_1,\ldots,s_n)$: for $d+1-(n+1)=d-n\leq j\leq -1$,
$\left(xg(tx^{-1}+s_1x^{-2}+\cdots+s_nx^{-n-1})\right)_j=g_jt+(g\ul{s})_j=0$ since $j\leq -1$ implies that $g_j=0$ and $g\in\Ann(s)$. The result now follows by induction on $k$.
\epr

\begin{lemma} \label{noreciprocals} Let $n\geq 2$ and $g\in \Ann^\bullet(s)$.

(i) If $h\in \Ann(s_2,\ldots,s_n)\setminus\Ann(s)$, then $\deg(g)\geq n-\deg(h)$.

(ii)  If $f\in\Ann(s)$,  $x|f$ and  $f/x\not\in\Ann(s^{(n-1)})$, then  $\deg(g)\geq n+1-\deg(f)$.
\end{lemma}
\begin{proof}  (i) Let $d=\deg(g)$. Then  $(g\ul{s})_j=0$ for $d-n\leq j\leq -1$ and so we can write $g\ul{s}=G+cx^{d-n-1}+P$ where $G\in D[x^{-1}]$, 
$G_j=0$ for $d-n-1\leq j\leq -1$, $c\in D$ and $P\in D[x]$.
Put $s'=(s_2,\ldots,s_n)$ and $e=\deg(h)$. Then we have $(h\ul{s'})_j=0$ for $e-n\leq j\leq -2$ and $a=(h\ul{s'})_{-1}\neq 0$.
As $\ul{s}=x^{-1}\ul{s'}+s_1x^{-1}$, we get
$$(h\ul{s})_j=\left\{   \begin{array}{ll} 
             0      & \mbox{ if } e-n-1\leq j\leq -3\\
             a      & \mbox{ if } j=-2\\
             b=(h\ul{s'})_0+s_1f_0 & \mbox{ if }j=-1
                                  \end{array}
          \right. $$
and so $h\ul{s}=H+ax^{-2}+bx^{-1}+Q$ where
$H\in D[x^{-1}]$, $H_j=0$ for $e-n-1\leq j\leq -1$ and $Q\in D[x]$. This gives 
$hP-gQ=gH-hG+gax^{-2}+gbx^{-1}-hcx^{d-n-1}\in D[x]$
and $$0=(gH)_{-2}-(hG)_{-2}+g_0a+h_{n+1-d}c.$$
If $d+e-n-1\leq -2$, then $n+1-d> e$ and so  this reduces to $0=ag_0$, for a contradiction. Hence $d+e-n-1\geq -1$, as required.

(ii) Proposition \ref{shift}(i) implies that $h=f/x\in\Ann(s_2,\ldots,s_n)$ and as $\Ann(s)\subseteq \Ann(s^{(n-1)})$, $h\not\in\Ann(s)$. Part (i) now implies that
 $\deg(g)\geq n-\deg(h)=n+1-\deg(f)$.
\end{proof}

The next corollary was stated without proof in \cite[Proof of Theorem 3.7]{Salagean} when $D$ is a field using different notation. 
\bl \label{x|min} Let  $\mu_0=0$ and $e>0$. If $g\in\Ann^\bullet(s)$ then $\deg(g)\geq n+1-\LC$. 
\el
\bpr We have $n+1-2\LC=e>0$. If $s$ is the all-zero sequence, $\mu=1$ and $\mu_0\neq 0$, so $s$ is not the all-zero sequence and $1\leq \LC\leq \frac{n}{2}$. Next we show that $\mu/x\not\in\Ann(s^{(n-1)})$. For suppose that $\mu/x\in\Ann(s^{(n-1)})$. Then
$\LC_{n-1}\leq \LC-1=\max\{\LC_{n-1},n-\LC_{n-1}\}-1$, which implies that $\LC_{n-1}\leq \LC-1=n-\LC_{n-1}-1$ 
i.e. $\ee_{n-1}>0$. Now Theorem \ref{bit} implies that $e=-\ee_{n-1}+1\leq 0$, which is a contradiction. Thus $\mu/x\not\in\Ann(s^{(n-1)})$ and the result follows from Lemma \ref{noreciprocals}(ii). 
\epr
\subsection{Two Corollaries, an Algorithm and a Theorem}
Recall that $e=\ee_n=n+1-2\LC$.
\bc \label{simple}(Cf. \cite[Theorem 3.8]{Salagean})  If $\mu_0=0$, let $q\in D[x],\ \deg(q)=e$,
$a\in D^\times$ and
$$\ol{\mu}^\bullet=\left\{   \begin{array}{ll} 
\ol{\mu}+a\cdot\ol{\mu}'       & \mbox{ if }  e\leq 0\\
            q\cdot\ol{\mu}+a\cdot\ol{\mu}' & \mbox{ otherwise.}
                                  \end{array}
          \right.
          $$ 
Then (i) $\ol{\mu}^\bullet\in\MR^\bullet(s)$ and (ii)
$\tilde\mu'\cdot\ol{\mu}^\bullet=\nabla_n$ if $e\leq 0$;
$\ol{\mu}\cdot\ol{\mu}^\bullet=-a\cdot \nabla_n$ if $e<0$.
\ec
\begin{proof}  Corollary \ref{nonzero} implies that $\mu'_0\neq 0$, so $\mu^\bullet_0\neq 0$. If $e\leq 0$, Part (i) follows from Proposition \ref{MRprop} with $f'=a$ and Part (ii) follows from Theorem \ref{numu}. Let $e>0$. We check that
$\mu^\bullet\in \Ann(s)$ directly from the definition.
Let $e+\LC-n=1-\LC\leq j\leq -1$. Then
$$(\mu^\bullet\cdot\ul{s})_j=(q\cdot\mu\cdot\ul{s})_{j}+a\cdot(\mu'\cdot\ul{s})_j=0$$
since (i) $1-\LC\leq j$ implies that  $\LC-n\leq j-e$ and $(q\cdot\mu\cdot\ul{s})_{j}=\sum_{k=0}^eq_k(\mu\cdot\ul{s})_{j-k}$, where $ j-e\leq j-k\leq j$ and (ii) $n'-\LC'\leq j$. We have $\deg(\mu^\bullet)=e+\LC$, which is minimal by Lemma \ref{x|min}. Finally, $\LC\geq 1$ implies that $e\leq n-\LC$ and hence by Lemma \ref{f2}(ii), $\mu^\bullet_2=q\mu_2+a\cdot\mu'_2$, so that $\ol{\mu}^\bullet=q\cdot\ol{\mu}+a\cdot\ol{\mu}'\in\MR^\bullet(s)$. For Part (ii), we easily get 
$\tilde{\mu}\cdot\ol{\mu}^\bullet=-a\cdot\tilde{\mu}'\cdot\ol{\mu}$ which is $-a\cdot\nabla_n$ by Theorem \ref{numu}. 
\end{proof} A trivial example of the first case is $n=1$ and $s_1\in D^\times$.
 We have $\ol{\mu}=(x,1)$, $\ol{\mu}'=(1,0)$ and $e=0$. According to the Corollary, we can take $\mu^\bullet=(x+a,1)$ where $a\in D^\times$. For the second case, let $n=4$ and $s=(0,1,0,0)$. Then $\ol{\mu}=(x^2,1)$, $\ol{\mu}'=(1,0)$ and $e=1$. According to the Corollary, $\mu^\bullet=x^e\ol{\mu}+a(1,0)=(x^3+a,x)$.

Lemma \ref{x|min} now yields the following value of $\LC^\bullet$. Recall that if $e<0$ then $\mu\in\Min(s)$ is unique.
\bc\label{LCbullet} For $s\in D^n$,
$$\LC^\bullet=\left\{   \begin{array}{ll} 
 \LC       & \mbox{ if }  e\leq 0 \mbox{ or }(e>0\mbox{ and }\mu_0\neq 0)\\
            n+1-\LC & \mbox{ otherwise.}
                                  \end{array}
          \right.
          $$ 
 \ec
We now have all the ingredients for an iterative algorithm for sequences over $D$ by taking $q=x^e$ and $a=1$ in Corollary \ref{simple}.
\begin{algorithm}\label{rewrite2} (Cf. \cite[Algorithm 3.2]{Salagean})

\begin{tabbing} Input:\ integer $n\geq 1$ and  $s=(s_1,\ldots,s_{n})\in D^n$.\\
Output:\  $\ol{\mu}^\bullet\in\MR^\bullet(s)$.\\\\
\{Algorithm \ref{MRa}\ $(\varepsilon=0,n,s,\ol{\mu},\ol{\mu}',e)$;\\\\
{\tt IF } $\mu_0\neq 0$\ \= \ {\tt THEN} $\ol{\mu}^\bullet:= \ol{\mu}$ \\
          \> \ {\tt ELSE}\ \{{\tt IF}\ $e\leq 0$\= \ {\tt THEN}\= \ $\ol{\mu}^\bullet:= \ol {\mu}+\ol{\mu}'$ \ {\tt ELSE}  $\ol{\mu}^\bullet:= x^{e}\cdot\ol{\mu}+\ol{\mu}'$\};\\
{\tt RETURN} (${\mu}^\bullet)$\}\\
\end{tabbing}
\end{algorithm}
\brs (Cf. \cite[Section III]{Salagean})  (i)  The correctness of Algorithm \ref{rewrite2} does not require us to  characterise elements of $\Ann^\bullet(s)$ of minimal degree. (ii) Theorem \ref{bit} uses any ${\mu}^{(n-1)}$; in particular, if $1\leq j\leq n$, $\ee_{j-1}\leq 0$ and ${\mu}^{(j-1)}_0=0$, we can replace $\ol{\mu}^{(j-1)}$ by $\ol{\mu}^{(j-1)}+\ol{\mu}'^{(j-1)}$.
 (iii)  Algorithm \ref{rewrite2} does not include any tests on $\mu'_0$. (iv) We can also use $\varepsilon =1$. 
 \ers
\begin{example} \label{bullet}We apply Algorithm \ref{rewrite2} to Example \ref{mrex}: $\ol{\mu}^{(8)}=(x^4+x^2+x,x^2+x+1)$, $\ol{\mu}'^{(8)}=(x^3+x^2+x+1,x)$  and $\ee_8=1$, so that $\ol{\mu}^{(8)\bullet}=x\ol{\mu}^{(8)}+\ol{\mu}^{(8')}=(x^5+x+1,x^3+x^2)$. 
\end{example}
The following theorem is our analogue of Proposition \ref{MRprop} and Theorem \ref{numu}.
\begin{theorem} \label{dge0} (Cf. \cite{Salagean}) Let $s\in D^n$ and $f'\in D[x]$. \begin{tabbing}
(i) \= If   $f'=0$ or $\deg(f')\leq -e$ and $\mu_0+ f'_0\cdot \mu'_0\neq 0$ then (a) $\ol{\mu}^\bullet=\ol{\mu}+f'\cdot \ol{\mu}'\in\MR^\bullet(s)$\\
\>and (b) $\tilde{\mu}'\cdot \ol{\mu}^{\bullet}=\nabla_n$.\\

(ii) \> If $e>0$ and $\mu_0=0$, let $s_{n+1}$ be such that  $\Delta_{n+1}=\Delta_{n+1}(\mu,t)\neq 0$, where\\
\>$t=(s_1,\ldots,s_n,s_{n+1})$. If $f'=0$ or $\deg(f')\leq e-1$ then\\\\

\> (a) $\ol{\mu}^{\bullet}=\ol{\mu}^{(n+1)}+f'\cdot\ol{\mu}\in \MR^\bullet(t)\subseteq\MR^\bullet(s)$\\\\
 
\> (b)  $\tilde{\mu}\cdot \ol{\mu}^{\bullet}=\nabla_{n+1}(t)$.
\end{tabbing}
\end{theorem}
\begin{proof} As in Proposition \ref{simple}, we know that $\mu'_0\neq 0$ by Proposition \ref{fg}, so the first part of (i) is clear from Proposition \ref{MRprop} and Theorem \ref{numu}, and $\tilde{\mu}'\cdot \ol{\mu}^{\bullet}=\tilde{\mu}'\cdot \ol{\mu}$. 
To prove Part (ii), first note that by Theorem \ref{bit},  $\mu^{(n+1)}=\Delta'_{n+1}\cdot x^{e}\mu-\Delta_{n+1}\cdot\mu'\in \Min(t)$ with  $\deg(\mu^{(n+1)})=n+1-\LC$ and $\Delta'_{n+1}=\Delta_{n'+1}$. Further,
$\mu^{(n+1)}_0=\Delta_{n+1}\cdot\mu'_0\neq 0$ and $\Ann(t)\subseteq\Ann(s)$, so Lemma \ref{x|min} implies that $\ol{\mu}^{(n+1)}\in \MR(t)\cap\MR^\bullet(s)$.  
We have $\ee_{n+1}=-e+1\leq 0$, so    $\ol{\mu}^{(n+1)}+f'\cdot\ol{\mu}'^{({n+1})}\in\MR(t)$ by Proposition \ref{MRprop}. Also, $(n+1)'=n$ since $e>0$. We conclude that  ${\mu}^{\bullet}_0=
{\mu}^{(n+1)}_0\neq 0$, 
$\ol{\mu}^{\bullet}=\ol{\mu}^{(n+1)}+f'\cdot\ol{\mu}
\in\MR^\bullet(t)$ and
$$\tilde{\mu}\cdot \ol{\mu}^{\bullet}=\tilde{\mu}'^{({n+1})}\cdot \ol{\mu}^{\bullet}=\tilde{\mu}'^{({n+1})}\cdot \ol{\mu}^{(n+1)}
=\nabla_{n+1}.$$
\epr
In Example \ref{bullet},  $s_9=0$ forces $\Delta_9=1$. We obtain $\ol{\mu}^{(9)}=\ol{\mu}^{\bullet(8)}$ as before (with $f'=0$). However, since $e=1$, we may also take $f'=1$, and we obtain $\ol{\mu}^{(9)}+\ol{\mu}^{(8)}=\ol{\mu}^{(8)\bullet}+\ol{\mu}^{(8)}= (x^5+x^4+x^2+1,x^3+x+1)\in\MR^\bullet(s)$ too.  This is equivalent to taking $f'=x+1$ in Corollary \ref{simple}.
\brs (i) We could replace $\mu'$ by any element of $\Min(s')$. (ii)  By construction, $\mu^\bullet=q\mu-\Delta\cdot\mu'$ where
$q=\Delta'\cdot x^e+f'$ and $\deg(q)=e$.
\ers
\subsection{A Characterisation} In this subsection, we give an analogue of Theorem \ref{MRthm} for $\Ann^\bullet(s)$ when $\mu$ is unique and $\mu_0=0$: we  characterise $\{f\in\Ann^\bullet(s):\  \deg(f)\mbox{ is minimal}\}$ in this case.  This was stated without proof for $D=\F$ in \cite[proof of Theorem 3.7(iii)]{Salagean}.  We also describe $f_2(x)=\sum_{j=0}^{\deg(f)-1} (f\cdot\ul{s})_j\ x^j$. (The case $e\leq 0$ is an easy consequence of Theorem \ref{MRthm}.)
First a simple but useful proposition.
\bp \label{converse} Let $f,r\in D[x]$ and $n-1\geq k=\deg(f)-\deg(r)\geq 1$. If  $f, f+r\in \Ann(s)$ then $r\in\Ann(s^{(n-k)})$.
\ep
\bpr As $n-k\geq 1$, $\Ann(s^{(n-k)})$ is well-defined. Let $d=\deg(f+r)=\deg(f)=\deg(r)+k$ and $j$ satisfy 
 $\deg(r)-n+k\leq j\leq -1$.  Also, $d-n= \deg(r)-(n-k)$ and so 
 $(r\cdot\ul{s})_j=((f+r)\cdot\ul{s})_j-(f\cdot\ul{s})_j=0$ and hence $r\in\Ann(s^{(n-k)})$.
\end{proof}

It is convenient  to have the notion of a 'jump point'. 
\bd [Jump point] Let $n\geq 2$, $s\in D^n$ and $2\leq j\leq n$. We say that $j$ is a jump point of $s$ if $\LC_j>\LC_{j-1}$ and write $\J(s)$ for the set of jump points of $s$. 
\ed
We do not assume that $\J(s)\neq \emptyset$. 
Evidently, the following are equivalent: (i) $j\in\J(s)$  (ii) $\Delta_j\neq 0$ and $\ee_{j-1}>0$ (iii)  $j'=j-1$ (iv) $\LC_j=j-\LC_{j-1}>\LC_j$.  In Theorem \ref{dge0}, $s_{n+1}$ was chosen so that $n+1$ is a jump point of $(s_1,\ldots,s_{n+1})$. 

\bt \label{char} (Cf. \cite[Theorem 3.7(iii)]{Salagean}) Let $n\geq 2$, $s\in D^n$, $\mu_0=0$, $e>0$ and $\ell=\mu_\LC^{e+1}$. Then $f\in\Ann^\bullet(s)$ has minimal degree if and only if  
 $$\deg(f)= n+1-\LC,\ n'\geq 1\mbox{ and }\ell\cdot\ol{f}=q\cdot \ol{\mu}+\ol{r}$$ where $ q=(\ell\cdot f)\mydiv\mu$,  $r=(\ell\cdot f) \bmod \mu\in\Min(s')$ and $r_0\neq 0$.  In particular, if $D$ is a field  then ${f}=q\cdot {\mu}+{r}$ where $\deg(q)=e$ and ${r}=a\cdot\ol{\mu}'$ for some $a\in D^\times.$
 \et
\bpr  Corollary \ref{simple} implies that $q\ol{\mu}+\ol{r}\in\MR^\bullet(s)$ and hence so is $\ol{f}$.
Conversely, let $f$ be as stated. Then $\deg(f)= n+1-\LC_n$ by  Corollary \ref{LCbullet}. 
 Since $\mu_0=0$, $\LC\geq 1$ and $e>0$ implies that   $\LC\leq \frac{n}{2}$.
Pseudo-dividing $\ell\cdot{f}$ by $\mu$ (as in  \cite[Algorithm R, p. 407]{K2}) gives 
$\ell\cdot f=q\mu+r$ with $\deg(q)=n+1-2\LC=e>0$ and $r=0$ or $\deg(r)<\LC$. Since $r_0=\ell f_0\neq 0$, $0\leq d=\deg(r)<\LC$.
 Let $k=\LC-d\geq 1$. Then $n-k\geq 1$ since $n\geq 2\LC>\LC\geq \LC-d=k$ and $s^{(n-k)}$ is well-defined. By Proposition \ref{converse}, $r\in\Ann(s^{(n-k)})$.
 
First we  show that $n-k=n'$. Since  $d<\LC$ and $n'$ is maximal by \cite[Proposition 4.1]{N95b},  $n-k\leq n'$. In particular, $n'\geq 1$ and we are done if $n'=1$.  Suppose that $n-k< n'$. 
If $\LC'\leq d+1$ then 
\begin{eqnarray}\label{firstcase}
n-\LC+\LC'\leq n-\LC+d+1=n-k+1\leq n'=\LC+\LC'-1
\end{eqnarray}
and $e\leq 0$. Similarly, $\LC'=d$ is impossible and so $d<\LC'$. To treat this case, we generalize (\ref{firstcase}) by partitioning $[\LC_1,\LC')$
into  $[\LC_1,\LC_{j_1}),\ldots, [\LC_{j_{p-1}},\LC_{j_p})$ for some integer $p$. There is a unique integer $t\geq 2$ with $d\in[\LC_{j_{t-1}},\LC_{j_t})$, where $j_{t_1}=1$. Since $d<\LC_{j_t}$, we must have $n-k\leq j_t-1$ as before and therefore
$$n-\LC+\LC_{j_{t-1}}\leq n-\LC+d=n-k\leq j_t-1=\LC_{j_{t}}+\LC_{j_{t-1}}-1\leq \LC+\LC_{j_{t-1}}-1$$
which implies that $e\leq 0$. We conclude that $n-k=n'$
and  $\LC'\leq d$ as $r\in\Ann(s')$. If  $\LC'<d$, then 
 $(n-\LC)+\LC'+1\leq (n-\LC)+d=n-k= n'=\LC+\LC'-1$
and $e\leq -1$, so that $d=\LC'$.

Finally, if $D$ is a field, we have shown that $f=q\cdot\mu+r$ where $r\in\Min(s')$ and $r_0\neq 0$. Moreover, $\LC'\leq \LC-1=n'-\LC'$ so that $\ee_{n'}=n'+1-2\LC'\geq 1$ and $\mu'$ is the unique monic minimal polynomial of $s'$ by \cite[Corollary 3.27 or Theorem 4.16]{N95b}. Hence $r=a\mu'$ for some $a\in D^\times$.
\epr
\br \label{MRchar} In Theorem \ref{char}, we have $\deg(q)+\deg(\mu)=n+1-\LC\leq n$ since $\LC\geq 1$ and it follows easily from \cite[Corollary 3.25]{N95b} (which does not require $D$ to be a field),
that $\ell\cdot\ol{f}=(q\mu+r,(q\mu+r)_2)=q\ol{\mu}+\ol{r}$.
\er

\noindent {\bf Example \ref{mrex}(cont.)}  From Theorem \ref{char} and Remark \ref{MRchar}, $$\MR^\bullet(s)=\{(x^5+x+1,x^3+x^2),(x^5+x^4+x^2+1,x^3+x+1)\}.$$

\begin{center}
{\bf Appendix}
\end{center}
We give an alternative proof of Lemma \ref{noreciprocals} using reciprocals of annihilators and apply it to the complexity of reversed sequences. This generalises  Theorem 3 of \cite[p. 149]{IY} (which was proved using using Hankel matrices over a field) to a factorial domain. 

Let $\tilde{s}$ denote the reverse of $s\in D^n$ and $\tilde{\LC}=\LC(\tilde{s})$. We assume that $\mu^{(n)}\in \Min(s)$.
We begin with a basic result on annihilators and reciprocals.
 
 \begin{proposition} \label{haha} Let  $f\in\Ann(s)$ and  $d=\deg(f)$ and let $k\geq 0$ be the highest power of $x$ dividing $f$. Then $f^\ast\in \Ann(s_n,\ldots,s_{k+1})$. In particular, if $x\nmid f$ then $f^\ast\in \Ann(\tilde{s})$.
\end{proposition}
\begin{proof}  First assume that $k=0$.
For $d-n\leq j\leq -1$, 
$$(f^\ast\ul{\tilde{s}})_j
=(f(x^{-1})\ul{\tilde{s}})_{j-d}=(f\cdot(s_nx+\cdots+s_1x^n))_{d-j}=(x^{n+1}f\ul{s})_{d-j}=(f\ul{s})_{d-j-n-1}=0$$ 
since $d-n\leq d-j-n-1\leq-1$. In general, let $f'=f/x^k$. Then $f'\in \Ann(
s_{k+1},\ldots,s_n)$ by Proposition \ref{shift}(i) and
$f^\ast=f'^\ast\in \Ann(s_n,\ldots,s_{k+1})$ by the case $k=0$.
\end{proof}

Next we give an alternative proof of Lemma  \ref{noreciprocals} using reciprocals of annihilators.
\begin{lemma} \label{Dbound} Let $n\geq 2$, $f,g\in\Ann(s)$. If $x|f$,  $f/x\not\in\Ann(s^{(n-1)})$ and $x\nmid g$, then $\deg(g)\geq n+1-\deg(f)$.
\end{lemma}
\begin{proof}  Since $g_0\neq 0$, $g\neq 0$ and $g^\ast\in\Ann(\tilde{s})$ by Proposition \ref{haha}(ii). Let $k$ be the highest power of $x$ dividing $f \in\Ann(s)$ and $\hat{f}=f/x^k$, so that $\deg(\hat{f}^\ast)=\deg(\hat{f})=\deg(f)-k$.   Firstly, $\hat{f}\in\Ann(s_{k+1},\ldots,s_n)$ by Proposition \ref{haha}(i) and $x\nmid \hat{f}$, so $\hat{f}^\ast\in\Ann(s_n,\ldots, s_{k+1})$ by Proposition \ref{haha}(ii).
But $\hat{f}^\ast\not\in\Ann(s_n,\ldots, s_{k})$, for otherwise $\hat{f}\in\Ann(s_k,\ldots,s_n)$ and hence $f/x=x^{k-1}\hat{f}\in\Ann(s)$ by Proposition \ref{haha}.  This is impossible since $\Ann(s)\subseteq\Ann(s^{(n-1)})$. 
We can now apply Lemma \ref{elegantproof} to $\hat{f}^\ast$ and $g^\ast$, yielding
$\deg(g^\ast)\geq n-k+1-\deg(\hat{f}^\ast)= n+1-\deg(f)$ i.e. $\deg(g)\geq n+1-\deg(f)$.
\end{proof}

\bl \label{geewiz} (i) If  $\mu_0\neq  0$ then $\tilde{\LC}\leq \LC$ and (ii) if $\mu_0=0$ then $\tilde{\LC}\geq n+1-\LC$.
\el
\bpr If  $\mu_0\neq  0$ then $\mu^\ast\in \Ann(\tilde{s})$
by Proposition \ref{haha}(ii), so $\tilde{\LC}\leq \deg(\mu^\ast)=\deg(\mu)=\LC$. Suppose now that $\mu_0=0$ and let $x^k|\mu$ with $k\geq 1$ maximal and put $g=\mu/x^k$. Then
$g_0\neq 0$, $\deg(g)=\LC-k$ and  $g\in\Ann(s_{k+1},\ldots,s_n)$ by Proposition \ref{haha}(ii). 
Hence ${g}^\ast\in \Ann(s_n,\ldots,s_{k+1})$. However, $g^\ast\not\in \Ann(s_k,\ldots,s_n)$, for otherwise $g\in \Ann(s_k,\ldots,s_n)$ and $\mu/x=x^{k-1}g\in \Ann(s)$ which  contradicts the minimality of $\mu$. We now apply Lemma \ref{elegantproof} to $g^\ast$:
$$\tilde{\LC}\geq \tilde{\LC}_{n-k-1}\geq n-k+1-\deg({g}^\ast)=n+1-{\LC}.$$
\epr
\begin{example} Let $D=\F_2$ and $s=(1,1,0,0)$. One easily checks that $x^2$ is the unique minimal polynomial.  By Lemma \ref{geewiz}, $\tilde{\LC}\geq 3$. In fact  
$\mu^{(3)}=x^3+x^2$ and so $\tilde{\LC}=3$.  Thus the bound of Lemma \ref{geewiz} is tight.
\end{example}

We will now show that in the special case $n=2\LC$ that 
$\LC=\LC+1$ when $D$ is factorial and $\mu^{(n)}_0=0$. Recall that if $D$ is a factorial domain and $n\geq 2\LC$ then $s$ has a unique minimal polynomial.
\bc (Cf. \cite[Theorem 3]{IY}). Let $D$ be a factorial domain, $n=2\LC$ and $\mu$ be the unique minimal polynomial of $s$. Then
$$\tilde{\LC}=\left\{   \begin{array}{ll} 
 \LC       & \mbox{ if }  \mu_0\neq 0\\
            \LC+1 & \mbox{ if }  \mu_0=0.
                                  \end{array}
          \right.
          $$ 
\ec
\bpr (i) We have  $\tilde{\LC}\leq \LC$ by Lemma \ref{geewiz}(i).  Let $\tilde{\mu}\in \Min(\tilde{s})$. If $\tilde{\mu}_0\neq 0$, $\tilde{\mu}^\ast\in \Ann(s)$, $\LC=\tilde{\tilde{\LC}}\leq \tilde{\LC}$ and we are done. If $\tilde{\mu}_0=0$, Lemma \ref{geewiz}(ii)  implies that $\LC\geq n+1-\tilde{\LC}=2\LC+1-\tilde{\LC}$. Hence
$\tilde{\LC}\geq \LC+1\geq \LC$ and $\tilde{\LC}=\LC$.

(ii) Suppose that $\mu_0=0$. Let $f$ be such that $f\in \Ann(s)$, $f_0\neq 0$ and $\deg(f)$ is minimal. Then $\deg(f)= n+1-\LC$ by Lemma 
\ref{noreciprocals}  and $\tilde{\LC}\leq n+1-\LC\leq \LC+1$ by Proposition \ref{haha} and the fact that $n=2\LC$. 
We now show that $\tilde{\LC}\leq \LC$ is impossible. 

Let $\tilde{f}\in \Min(\tilde{s})$. If $\tilde{f}_0\neq 0$, then $\tilde{f}^\ast\in\Ann(s)$ and so
$\LC\leq \deg(\tilde{f})=\tilde{\LC}$. We cannot have $\LC=\tilde{\LC}$, for then $\tilde{f}^\ast=\mu$ and $\mu_0\neq 0$. Hence $\LC< \tilde{\LC}$ and we cannot have $\tilde{\LC}\leq \LC$. By the first paragraph, $\tilde{\LC}=\LC+1$.
On the other hand, if $\tilde{f}_0=0$, then by Lemma \ref{geewiz},
$\tilde{\LC}\geq n+1-{\LC}=\LC+1$. Thus in either case,
$\tilde{\LC}\geq \LC+1$ and so $\tilde{\LC}= \LC+1$.
\epr
\begin{center}
{\bf Typographical Error}
\end{center}
In \cite{N10a}, immediately after Definition 3.3, $\Delta(f,s)=\sum_{k=0}^d f_k\ s_{n-d+k}$.

\end{document}